\documentclass[12pt,twoside,singlespace]{article}
\usepackage{array, paralist, enumerate, amsmath, amsfonts, amssymb, amscd, color, mathrsfs,comment}
\usepackage{amsthm} % place after to make qedhere work

\usepackage{geometry}
\usepackage{framed}
\usepackage{hyperref}
\usepackage{graphicx}
\usepackage{epstopdf}
\usepackage[all,cmtip]{xy}
\usepackage{tikz}
\usepackage{tkz-graph}
\usetikzlibrary{arrows,%
                shapes,positioning}

\definecolor{DarkBlue}{rgb}{0,0,0.8} 
\definecolor{DarkGreen}{rgb}{0,0.5,0.0} 
\definecolor{DarkRed}{rgb}{0.9,0.0,0.0} 

\usepackage[T1]{fontenc}
\usepackage[latin1]{inputenc}

\numberwithin{equation}{section}
\newtheorem{thm}[equation]{Theorem}
\newtheorem{lem}[equation]{Lemma}
\newtheorem{cor}[equation]{Corollary}
\newtheorem{prop}[equation]{Proposition}

\theoremstyle{definition}
\newtheorem{definition}[equation]{Definition}
\newtheorem{ex}[equation]{Example}

\newtheorem{construction}[equation]{Construction}

\newcommand{\ZZ}{\mathbf{Z}}

\newcommand{\on}{\operatorname}

\newcommand{\val}{\on{Val}}

\newcommand{\wt}{\on{wt}}

\renewcommand{\vec}[1]{\mathbf{#1}}

\title{Structural Theory of $2$-d Adinkras}
\author{Kevin Iga and Yan X Zhang}

\begin{document}

\pagestyle{plain}

\maketitle

\begin{abstract}
Adinkras are combinatorial objects developed to study ($1$-dimensional) supersymmetry representations. Recently, \emph{$2$-d Adinkras} have been developed to study $2$-dimensional supersymmetry.  In this paper, we classify all $2$-d Adinkras, confirming a conjecture of T. H\"ubsch.  Along the way, we obtain other structural results, including a simple characterization of H\"ubsch's \emph{even-split doubly even codes}. 
\end{abstract}

\section{Introduction}
Adinkras (in this paper, called \emph{$1$-d Adinkras}) were introduced in \cite{d2l:first} to study representations of the super-Poincar\'e algebra in one dimension.  There have been a number of developments that have led to the classification of $1$-d Adinkras.\cite{d2l:graph-theoretic,d2l:decodes,d2l:omni,d2l:topology, dil:cohomology,zhang:adinkras}  Based on the success of this program, there have been a few recent approaches to using Adinkra-like ideas to study the super-Poincar\'e algebra in two dimensions.  This has led to the development of $2$-d Adinkras.\cite{gates:dimensional_extension,hubsch:weaving}

In this paper, we characterize $2$-d Adinkras, guided by the approach and conjectures set forth in \cite{hubsch:weaving}. We begin in Section~\ref{sec:prelim} by recalling the definition of ($1$-d) Adinkras and some of their features, reviewing the \emph{code} associated with an Adinkra\cite{d2l:omni} and the concept of \emph{vertex switching}.\cite{dil:cohomology,zhang:adinkras} As this paper is a mostly self-contained work of combinatorial classification, we do not discuss (or require from the reader) the physics and representation theory background relating to $1$-d Adinkras; the interested reader may see Appendix~\ref{app:repn} and the aforementioned references for more information along these lines. Instead, Section~\ref{sec:prelim}'s goal is to provide the minimum background to understand and manipulate Adinkras as purely combinatorial objects. 

Then, Sections~\ref{sec:2d}--\ref{sec:code2d} discuss $2$-d Adinkras: the definition, some basic constructions, and characterizing their codes. In Section~\ref{sec:quotient}, we prove the main theorem, which is H\"ubsch's conjecture:\footnote{The formulation in \cite{hubsch:weaving} is slightly different: see Appendix~\ref{app:repn} for details.}

\begin{thm}
\label{thm:main}
Let $A$ be a connected $2$-d Adinkra.  Then there exist $1$-d Adinkras $A_1$ and $A_2$ so that
\[A\cong F(A_1\times A_2)/\sim, \]
where $F$ is a vertex switching and  $\sim$ is described by an action of a subgroup of $\ZZ_2^n$.
\end{thm}

Finally, Section~\ref{sec:structure}, guided by the main theorem, summarizes the basic structure of $2$-d Adinkras, including a (computable but impractical due to combinatorial explosion) scheme to generate all $2$-d Adinkras. We end with some remarks in Section~\ref{sec:conclusion}.

\section{Preliminaries}
\label{sec:prelim}

\subsection{$1$-d Adinkras}
\label{sec:1d}
\emph{Adinkras} in \cite{d2l:first,d2l:graph-theoretic,zhang:adinkras} will be referred to as \emph{$1$-d Adinkras} in this paper, since they relate to supersymmetry in $1$ dimension. In this section, we  review a definition of $1$-d Adinkras and give some tools from previous work on their structural theory. The material in this section is mainly found in \cite{d2l:omni,zhang:adinkras}, with minor paraphrasing. 

\begin{definition}[$1$-d Adinkras]
Let $n$ be a non-negative integer.  An \emph{$1$-d Adinkra} with $n$ colors is $(V,E,c,\mu,h)$ where: 

\begin{enumerate}
\item $(V,E)$ is a finite undirected graph (called the \emph{underlying graph} of the Adinkra) with vertex set\footnote{In \cite{d2l:first,d2l:graph-theoretic}, there is also a bipartition of the vertices, where some vertices are represented by open circles and called bosons, and other vertices are represented by filled circles and called fermions.  This is not necessary to include in our definition, because the bipartition can be obtained directly by taking the grading $h$ modulo $2$, which is a bipartition by property 4 below.} $V$ and edge set $E$.
\item $c:E\to [n] := \{1,\ldots,n\}$ is a map called the \emph{coloring}. We require that for every $v\in V$ and $i \in [n]$, there exists exactly one $w\in V$ so that $(v,w)\in E$ and $c(v,w)=i$. We also require that every two-colored simple cycle be of length $4$ (A \emph{simple} cycle is one which does not repeat vertices other than the starting vertex; A \emph{two-colored} cycle is one where the set of colors of the edges has cardinality $2$).
\item $\mu:E\to \ZZ_2=\{0,1\}$ is a map called the \emph{dashing}.  The \emph{parity} of $\mu$ on a cycle given by vertices $(v_0,\ldots,v_k)$ is defined as the sum
\[\sum_{i=0}^{k-1}\mu(v_i,v_{i+1})\pmod{2}.\]
We require that the parity of $\mu$ on every two-colored simple cycle to be odd. Such a dashing $\mu$ is called \emph{admissible}.
\item $h:V\to\ZZ$ is a map called the \emph{grading}. We require that if $(v,w)\in E$, then $|h(v)-h(w)|=1$. Equivalently, $h$ provides a height function that makes $(V,E)$ into the Hasse diagram of a ranked poset.
\end{enumerate}
\end{definition}

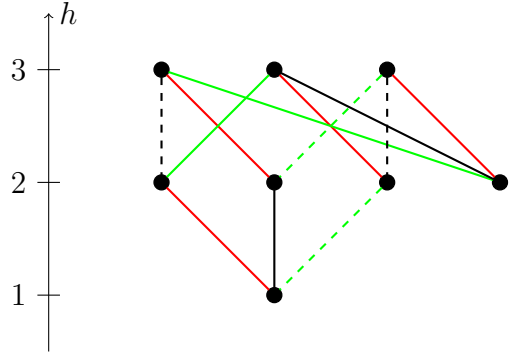
\begin{figure}
\begin{center}
\begin{tikzpicture}[scale=0.15]
\SetVertexSimple[MinSize=5pt]
\SetUpEdge[labelstyle={draw}]
\Vertex[x=0,y=0]{A}
\Vertex[x=0,y=10]{B}
\Vertex[x=0,y=20]{C}
\Vertex[x=20,y=10]{D}
\Vertex[x=-10,y=10]{E}
\Vertex[x=10,y=10]{F}
\Vertex[x=-10,y=20]{G}
\Vertex[x=10,y=20]{H}
\Edge [color=red](G)(B)
\Edge[color=red](D)(H)
\Edge[color=red](C)(F)
\Edge[color=red](E)(A)
\Edge[color=green](D)(G)
\Edge[color=green, style=dashed](H)(B)
\Edge[color=green](C)(E)
\Edge[color=green, style=dashed](F)(A)
\Edge(D)(C)
\Edge[style=dashed](H)(F)
\Edge[style=dashed](G)(E)
\Edge(B)(A)
\draw [->] (-20,-5) -- (-20,25);
\draw (-19,0) -- (-21,0) node [align=right, left] {$1$};
\draw (-19,10) -- (-21,10) node [align=right, left] {$2$};
\draw (-19,20) -- (-21,20) node [align=right, left] {$3$};
\node [right] at (-20,25) {$h$};
\end{tikzpicture}
\caption{Example of a $1$-d Adinkra with $3$ colors. The coloring is represented by colors on the edges. The dashing is represented by having a dashed edge if $\mu(e)=1$ and a solid edge if $\mu(e)=0$. The grading is represented by the vertical height as indicated on the axis on the left.
\label{fig:1d-examples}}
\end{center}
\end{figure}

Figure~\ref{fig:1d-examples} gives an example of a $1$-d Adinkra. 

\subsection{Structural Aspects of $1$-d Adinkras}
\label{sec:code}

 Let $A$ be a $1$-d Adinkra with $n$ colors, with vertex set $V$.  For all $i\in [n]$, define $q_i:V\to V$ such that for all $v\in V$, $q_i(v)$ is the unique vertex joined to $v$ by an edge of color $i$.  In \cite{d2l:omni}, it was shown that the map $q_i$ is a graph isomorphism (in fact, an involution) from the underlying graph of $A$ to itself which preserves colors. The $q_i$ commute with each other. These facts can be used to combine the $q_1,\ldots, q_n$ maps into an action of $\ZZ_2^n$ on the graph $(V,E)$ underlying the Adinkra in the following way:
\begin{definition}
The action of $\ZZ_2^n$ on the graph $(V,E)$ underlying the Adinkra is given on vertices by
\[(x_1,\ldots,x_n)v=q_1^{x_1}\circ\cdots\circ q_n^{x_n}(v).\]
\end{definition}

Intuitively, the action of a sequence of bits, for instance, $11001$, on a vertex is obtained by following edges with colors that correspond to $1$'s in the sequence (in this case, colors $1$, $2$, and $5$).  The fact that the $q_i$'s commute implies that the order of the colors does not matter.

The Adinkra $A$ is connected if and only if the $\ZZ_2^n$ action is transitive on the vertex set of $A$.  In this case the stabilizers of all vertices are equal (in general the stabilizers of two points in the same orbit are conjugate; here we know more since the group is abelian).  Define $C(A)$, the \emph{code of the Adinkra} $A$, to be this stabilizer.  This is a \emph{binary linear code} of length $n$ (i.e., a linear subspace of $\ZZ_2^n$). As these are the only types of codes we use, from now on we simply say \emph{code} to mean ``binary linear code.''

We call the elements of a code  \emph{codewords}. The \emph{weight} of a codeword $w$ is the number of $1$'s in the word. A code is called \emph{even} if all its codewords have even weight. A code is called \emph{doubly-even} if all its codewords have weight divisible by $4$. An example of a doubly-even code is the span $\langle 111100, 001111\rangle$, which has $2^2 = 4$ elements. An example of a code that is even but not doubly-even is the $1$-dimensional code $\langle 11 \rangle$.

Codes are surprisingly relevant to the structural theory of Adinkras; in fact, one should basically think of the underlying graph of an Adinkra as a doubly-even code, as we now see.

\subsection{Quotients}
We now know that the stabilizer of our action on the graph is a code. We can also go in the opposite direction: let $I^n$, the \emph{Hamming cube}, be the graph with $2^n$ vertices labeled by strings of length $n$ using the alphabet $\{0, 1\}$, with an edge between two vertices $v$ and $w$ if and only if they differ in exactly one place. There is a natural coloring on $I^n$: just color each edge by the coordinate where the two vertices differ.  Now, codes in $\ZZ_2^n$ act on $I^n$ by bitwise addition modulo $2$, and these are isomorphisms that preserve colors.  A natural operation to consider on a colored graph $\Gamma=(V,E,c)$ and a group $C$ acting on $V$ via graph isomorphisms that preserve colors is the \emph{quotient} $\Gamma/C$, where the vertices are defined to be orbits in $V/C$, and we have an $c$-colored edge $(v,w)$ if and only if there is at least one $c$-colored edge $(v',w') \in E$ with $v'$ in the orbit $v$ and $w'$ in the orbit $w$. 

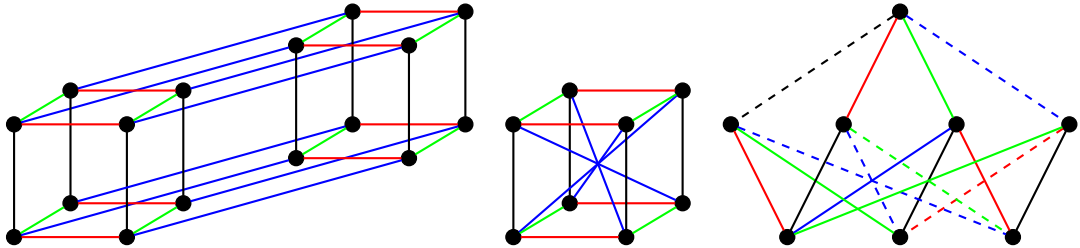
\begin{figure}[htb]
\begin{center}

\begin{tabular}{ccc}
\begin{tikzpicture}[scale=0.15]
\SetVertexSimple[MinSize=5pt]
\SetUpEdge[labelstyle={draw}]
\Vertex[x=0,y=0]{A}
\Vertex[x=0,y=10]{B}
\Vertex[x=10,y=0]{C}
\Vertex[x=5,y=3]{D}
\Vertex[x=15,y=13]{E'}
\Vertex[x=10,y=10]{F}
\Vertex[x=5,y=13]{G}
\Vertex[x=15,y=3]{H}
\Vertex[x=40,y=20]{A'}
\Vertex[x=40,y=10]{B'}
\Vertex[x=30,y=20]{C'}
\Vertex[x=35,y=17]{D'}
\Vertex[x=25,y=7]{E}
\Vertex[x=30,y=10]{F'}
\Vertex[x=35,y=7]{G'}
\Vertex[x=25,y=17]{H'}
\Edge[color=black](A)(B)
\Edge[color=black](A')(B')
\Edge[color=red](A)(C)
\Edge[color=red](A')(C')
\Edge[color=green](A)(D)
\Edge[color=green](A')(D')
\Edge[color=blue](A)(E)
\Edge[color=blue](A')(E')
\Edge[color=red](B')(F')
\Edge[color=green](B)(G)
\Edge[color=green](B')(G')
\Edge[color=blue](B)(H')
\Edge[color=blue](B')(H)
\Edge[color=black](C')(F')
\Edge[color=green](C)(H)
\Edge[color=green](C')(H')
\Edge[color=blue](C)(G')
\Edge[color=blue](C')(G)
\Edge[color=black](D)(G)
\Edge[color=black](D')(G')
\Edge[color=red](B)(F)
\Edge[color=red](D)(H)
\Edge[color=red](D')(H')
\Edge[color=black](C)(F)
\Edge[color=blue](D)(F')
\Edge[color=blue](D')(F)
\Edge[color=black](E)(H')
\Edge[color=black](E')(H)
\Edge[color=red](E)(G')
\Edge[color=red](E')(G)
\Edge[color=green](E)(F')
\Edge[color=green](E')(F)
\end{tikzpicture}
&
\begin{tikzpicture}[scale=0.15]
\SetVertexSimple[MinSize=5pt]
\SetUpEdge[labelstyle={draw}]
\Vertex[x=0,y=0]{A}
\Vertex[x=0,y=10]{B}
\Vertex[x=10,y=0]{C}
\Vertex[x=5,y=3]{D}
\Vertex[x=15,y=13]{E}
\Vertex[x=10,y=10]{F}
\Vertex[x=5,y=13]{G}
\Vertex[x=15,y=3]{H}
\Edge[color=black](D)(G)
\Edge[color=red](D)(H)
\Edge[color=black](E)(H)
\Edge[color=red](E)(G)
\Edge[color=blue](A)(E)
\Edge[color=blue](B)(H)
\Edge[color=blue](C)(G)
\Edge[color=blue](D)(F)
\Edge[color=green](A)(D)
\Edge[color=green](B)(G)
\Edge[color=green](E)(F)
\Edge[color=green](C)(H)
\Edge[color=black](A)(B)
\Edge[color=red](A)(C)
\Edge[color=red](B)(F)
\Edge[color=black](C)(F)
\end{tikzpicture}
&
\begin{tikzpicture}[scale=0.15]
\SetVertexSimple[MinSize=5pt]
\SetUpEdge[labelstyle={draw}]
\Vertex[x=15,y=20]{0000}
\Vertex[x=0,y=10]{0001}
\Vertex[x=10,y=10]{0010}
\Vertex[x=20,y=10]{0100}
\Vertex[x=30,y=10]{0111}
\Vertex[x=5,y=0]{0011}
\Vertex[x=15,y=0]{0101}
\Vertex[x=25,y=0]{0110}
\Edge[color=black](0100)(0101)
\Edge[color=red](0100)(0110)
\Edge[color=black](0111)(0110)
\Edge[color=red, style=dashed](0111)(0101)
\Edge[color=blue,style=dashed](0000)(0111)
\Edge[color=blue,style=dashed](0001)(0110)
\Edge[color=blue,style=dashed](0010)(0101)
\Edge[color=blue](0100)(0011)
\Edge[color=green](0000)(0100)
\Edge[color=green](0001)(0101)
\Edge[color=green](0111)(0011)
\Edge[color=green, style=dashed](0010)(0110)
\Edge[color=black, style=dashed](0000)(0001)
\Edge[color=red](0000)(0010)
\Edge[color=red](0001)(0011)
\Edge[color=black](0010)(0011)
\end{tikzpicture}
\end{tabular}
\caption{Left: the colored graph $I^4$. Middle: the quotient $I^4/\{0000,1111\}$. Right: as the code $\{0000, 1111\}$ is doubly-even, there exists an Adinkra with the quotient as its underlying graph by Theorem~\ref{thm:1d-quotients}. \label{fig:4cube folding}}
\end{center}
\end{figure}

\begin{thm}
\label{thm:1d-quotients}
$I^n/C$ is the colored graph of some $1$-d Adinkra if and only if $C$ is a doubly-even code.
\end{thm}

See \cite{d2l:omni} for the original proof of this result. See \cite{zhang:adinkras} for a more general treatment of quotienting by a code and a slightly extended correspondence\footnote{Note that the quotient $\Gamma/C$ does not necessarily retain nice properties of $\Gamma$; it does not even have to be a simple graph. It may also have edges with different colors between two vertices. Part of the work here is to show these pathologies do not happen when $C$ is a doubly-even code.} between graph properties of the quotient $I^n/C$ and properties of the code $C$.

Figure~\ref{fig:4cube folding} provides an example of a quotient of $I^4$ by a code that obeys this theorem. Encoded within the proof of Theorem~\ref{thm:1d-quotients} is the fact that if $C$ is a doubly-even code, then there exists an admissible dashing. A constructive proof of existence can be found in \cite{d2l:topology}. See \cite{zhang:adinkras} for an enumeration of all admissible dashings for any doubly-even code.

\subsection{Vertex Switching}
\label{sec:vertexswitch}
Vertex switching was first introduced in the context of Adinkras in \cite{d2l:first} and is more thoroughly set in its context in \cite{dil:cohomology,zhang:adinkras}.

\begin{definition}[Vertex switching]
Given an Adinkra $A$, and a vertex $v$ of $A$, we define \emph{vertex switching at $v$} to be the operation on $A$ that returns a new Adinkra $\bar{A}$ with the same vertices, edges, coloring, and grading but a new dashing $\bar{\mu}$ so that
\begin{equation}
\bar{\mu}(e)=\begin{cases}
1-\mu(e),&\mbox{if $e$ is incident to $v$}\\
\mu(e),&\mbox{otherwise.}
\end{cases}
\end{equation}
We leave to the reader to check that $\bar{\mu}$ is still an admissible dashing; since the vertices, edges, coloring, and grading(s) are the same, $\bar{A}$ remains an Adinkra. We also use a \emph{vertex switching of $A$} to refer to a composition of vertex switchings at various vertices of $A$. 
\end{definition}

\begin{figure}[htb]
\begin{center}
\begin{tabular}{cc}
\begin{tikzpicture}[scale=0.15]
\SetVertexSimple[MinSize=5pt]
\SetUpEdge[labelstyle={draw}]
\Vertex[x=0,y=0]{A}
\Vertex[x=0,y=10]{B}
\Vertex[x=10,y=0]{C}
\Vertex[x=5,y=3]{D}
\Vertex[x=15,y=13]{E}
\Vertex[x=10,y=10]{F}
\Vertex[x=5,y=13]{G}
\Vertex[x=15,y=3]{H}
\Edge[color=black,style=dashed](A)(B)
\Edge[color=red](A)(C)
\Edge[color=green](A)(D)
\Edge[color=black](D)(G)
\Edge[color=red](B)(F)
\Edge[color=green](B)(G)
\Edge[color=green,style=dashed](C)(H)
\Edge[color=red](D)(H)
\Edge[color=black](C)(F)
\Edge[color=black](E)(H)
\Edge[color=red,style=dashed](E)(G)
\Edge[color=green](E)(F)
\node[text width=1cm] at (19,15) {$v$};
\node[text width=1cm] at (10,-5){$A$};
\end{tikzpicture}
&
\begin{tikzpicture}[scale=0.15]
\SetVertexSimple[MinSize=5pt]
\SetUpEdge[labelstyle={draw}]
\Vertex[x=0,y=0]{A}
\Vertex[x=0,y=10]{B}
\Vertex[x=10,y=0]{C}
\Vertex[x=5,y=3]{D}
\Vertex[x=15,y=13]{E}
\Vertex[x=10,y=10]{F}
\Vertex[x=5,y=13]{G}
\Vertex[x=15,y=3]{H}
\Edge[color=black,style=dashed](A)(B)
\Edge[color=red](A)(C)
\Edge[color=green](A)(D)
\Edge[color=black](D)(G)
\Edge[color=red](B)(F)
\Edge[color=green](B)(G)
\Edge[color=green,style=dashed](C)(H)
\Edge[color=red](D)(H)
\Edge[color=black](C)(F)
\Edge[color=black,style=dashed](E)(H)
\Edge[color=red](E)(G)
\Edge[color=green,style=dashed](E)(F)
\node[text width=1cm] at (19,15) {$v$};
\node[text width=1cm] at (10,-5){$\bar{A}$};
\end{tikzpicture}
\end{tabular}
\caption{A vertex switching at $v$ turns the Adinkra $A$ on the left into the Adinkra $\bar{A}$ on the right.  The two Adinkras have the same dashing except precisely the edges that are incident to $v$.  Note that in both cases, each face of the cube has an odd number of dashed edges.\label{fig:vertexswitch}}
\end{center}
\end{figure}
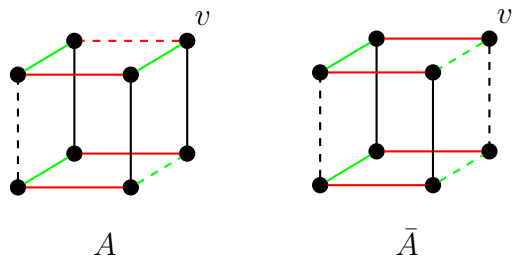

In \cite{douglas}, vertex switching was first applied to dashings in Adinkras from a point of view inspired by Seidel's \emph{two-graphs} \cite{seidel:survey}.\footnote{In Seidel's setting, \emph{vertex switching} switched the existence of edges, not the sign of edges; this can be seen as equivalent our definition applied to the complete graph. The type of vertex switching we do in this paper is sometimes called vertex switching on \emph{signed graphs} in literature for disambiguation.} An enumeration of vertex switching classes leading to counting the number of dashings of any $1$-d Adinkra can be found in \cite{zhang:adinkras}.

\section{$2$-d Adinkras}
\label{sec:2d}
Just as $1$-d Adinkras were used to study $1$-d supersymmetry, Gates and H\"ubsch developed $2$-d Adinkras to study $2$-d supersymmetry.\cite{gates:dimensional_extension,hubsch:weaving}  We use a definition here that is equivalent to the one found there.\footnote{The main notational difference is a kind of change of coordinates: there, nodes are labeled by \emph{mass dimension}, which is $h_L+h_R$, and \emph{spin}, which is $h_R-h_L$.  Mass dimension is the units of mass associated with the field, where $c=\hbar=1$ and spin is the eigenvalue of $x\partial_t+t \partial_x$.}

A $2$-d Adinkra is similar to a $1$-d Adinkra except that some colors are called ``left-moving'' and the other colors called ``right-moving''.  Edges are called ``left-moving'' if they are colored by left-moving colors, and right-moving otherwise.  Furthermore, there are two gradings, one that is affected by the left-moving edges and the other for the right-moving edges. More formally:
\begin{definition}[$2$-d Adinkras]
Let $p$ and $q$ be non-negative integers. A \emph{2-d Adinkra with $(p,q)$ colors} is a 1-d Adinkra $(V,E,c,\mu,h)$ with $p+q$ colors, and two grading functions $h_L:V\to \ZZ$ and $h_R:V\to \ZZ$ so that
\begin{itemize}
\item $h(v)=h_L(v)+h_R(v)$.
\item Let $e$ be an edge.  If $c(e)\le p$ then $e$ is called a \emph{left-moving edge}; if $c(e)>p$ then it is called a \emph{right-moving edge}. Similarly, the first $p$ colors are called \emph{left-moving colors} and the last $q$ colors are called \emph{right-moving colors}.
\item if $(v,w)$ is a left-moving edge, then $|h_L(v)-h_L(w)|=1$ and $h_R(v)=h_R(w)$.  If $(v,w)$ is a right-moving edge, then $|h_R(v)-h_R(w)|=1$ and $h_L(v)=h_L(w)$.
\end{itemize}
\end{definition}

\begin{figure}[htb]
\begin{center}
\begin{tikzpicture}[scale=0.05]
\SetVertexSimple[MinSize=5pt]
\node[text width=3cm] at (-10, 0) {(0,0)};
\node[text width=3cm] at (-35, 30) {(1,0)};
\node[text width=3cm] at (65, 30) {(0,1)};
\node[text width=3cm] at (-10, 60) {(1,1)};
\SetUpEdge[labelstyle={draw}]
\Vertex[x=0,y=0]{A}
\Vertex[x=-10,y=0]{H}
\Vertex[x=-35,y=30]{C}
\Vertex[x=-25,y=30]{B}
\Vertex[x=25,y=30]{D}
\Vertex[x=15,y=30]{E}
\Vertex[x=0,y=60]{G}
\Vertex[x=-10,y=60]{F}
\Edge[color=red](A)(C)
\Edge[color=red](B)(H)
\Edge[color=red](G)(E)
\Edge[color=red](F)(D)
\Edge[color=green](A)(D)
\Edge[color=green, style=dashed](E)(H)
\Edge[color=green](G)(B)
\Edge[color=green, style=dashed](F)(C)
\Edge[color=blue, style=dashed](C)(H)
\Edge[color=blue](B)(A)
\Edge[color=blue, style=dashed](G)(D)
\Edge[color=blue](F)(E)
\Edge[color=black, style=dashed](D)(H)
\Edge[color=black, style=dashed](A)(E)
\Edge[color=black](G)(C)
\Edge[color=black](B)(F)
\end{tikzpicture}
\caption{A $2$-d Adinkra with $(2,2)$ colors. The grading coordinates are given next to the nodes as $(h_L, h_R)$. \label{fig:2d-example}}
\end{center}
\end{figure}
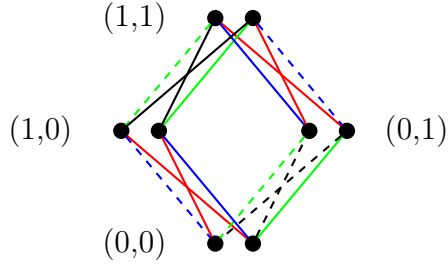

See Figure~\ref{fig:2d-example} for an example of a $2$-d Adinkra. The main goal of this paper is to follow the program set out in \cite{hubsch:weaving} and completely characterize $2$-d Adinkras. As a first step, we define the natural notion of \emph{products} in the following section.

\section{Products}
\label{sec:products}
One important way to produce a $2$-d Adinkra with $(p,q)$ colors is to take the product of two $1$-d Adinkras (one with $p$ colors, and the other with $q$ colors), using the following construction.

\begin{construction}
\label{const:product}
Let $p$ and $q$ be non-negative integers.  Let $A_1=(V_1, E_1, c_1, \mu_1,h_1)$ be a $1$-d Adinkra with $p$ colors and let $A_2=(V_2, E_2, c_2, \mu_2,h_2)$ be a $1$-d Adinkra $q$ colors.  We define the \emph{product} of these Adinkras $A_1\times A_2$ as the following 2-Adinkra with $(p,q)$ colors:
\[A_1\times A_2=(V,E,c,\mu,h_1,h_2)\]
where $V=V_1\times V_2$ and there are two kinds of edges in $E$:
\begin{itemize}
\item For every edge $e$ in $E_1$ connecting vertices $v$ and $w\in V_1$, and for every vertex $x\in V_2$, we have an edge in $E$ between vertices $(v,x)$ and $(w,x)$ in $V=V_1\times V_2$ of color $c_1(e)$ and dashing $\mu_1(e)$.
\item For every edge $e$ in $E_2$ connecting vertices $v$ and $w\in V_2$ and for every vertex $x\in V_1$, we have an edge in $E$ between vertices $(x,v)$ and $(x,w)$ in $V=V_1\times V_2$ of color $p+c_2(e)$ and dashing $\mu_2(e)+h_1(x)\pmod{2}$.
\end{itemize}
 See Figure~\ref{fig:product} for an example.
\end{construction}

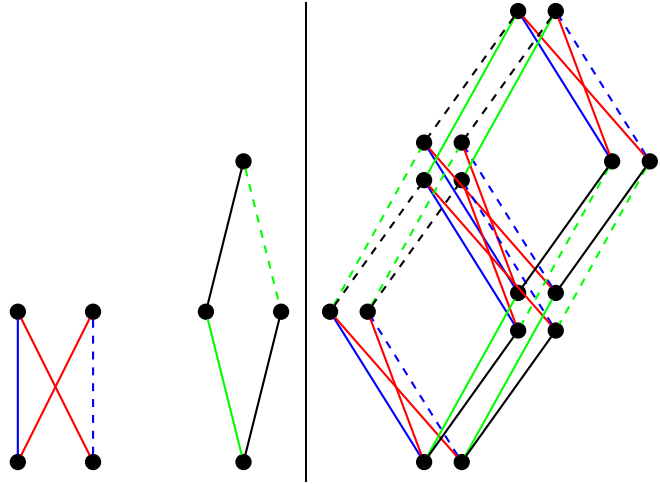
\begin{figure}
\begin{center}
\begin{tabular}{c|c}
\begin{tikzpicture}[scale=0.1]
\SetVertexSimple[MinSize=5pt]
\SetUpEdge[labelstyle={draw}]
\Vertex[x=0,y=0]{A}
\Vertex[x=10,y=0]{B}
\Vertex[x=0,y=20]{E}
\Vertex[x=10,y=20]{F}
\Edge[color=blue](A)(E)
\Edge[color=blue, style=dashed](B)(F)
\Edge[color=red](A)(F)
\Edge[color=red](B)(E)

\Vertex[x=30,y=0]{A'}
\Vertex[x=30,y=40]{B'}
\Vertex[x=25,y=20]{E'}
\Vertex[x=35,y=20]{F'}
\Edge[color=green](A')(E')
\Edge[color=green, style=dashed](B')(F')
\Edge[color=black](A')(F')
\Edge[color=black](B')(E')
\end{tikzpicture}
&
\begin{tikzpicture}[scale=0.05]
\SetVertexSimple[MinSize=5pt]
\SetUpEdge[labelstyle={draw}]
\Vertex[x=0,y=0]{BA}
\Vertex[x=-10,y=0]{AA}

\Vertex[x=-35,y=40]{EA}
\Vertex[x=-25,y=40]{FA}

\Vertex[x=15,y=35]{AF}
\Vertex[x=25,y=35]{BF}
\Vertex[x=15,y=45]{AE}
\Vertex[x=25,y=45]{BE}

\Vertex[x=-10,y=75]{EF}
\Vertex[x=0,y=75]{FF}
\Vertex[x=-10,y=85]{EE}
\Vertex[x=0,y=85]{FE}

\Vertex[x=40,y=80]{AB}
\Vertex[x=50,y=80]{BB}
\Vertex[x=15,y=120]{EB}
\Vertex[x=25,y=120]{FB}

\Edge[color=blue](AA)(EA)
\Edge[color=blue, style=dashed](BA)(FA)
\Edge[color=red](AA)(FA)
\Edge[color=red](BA)(EA)
\Edge[color=blue](AB)(EB)
\Edge[color=blue, style=dashed](BB)(FB)
\Edge[color=red](AB)(FB)
\Edge[color=red](BB)(EB)
\Edge[color=blue](AE)(EE)
\Edge[color=blue, style=dashed](BE)(FE)
\Edge[color=red](AE)(FE)
\Edge[color=red](BE)(EE)
\Edge[color=blue](AF)(EF)
\Edge[color=blue, style=dashed](BF)(FF)
\Edge[color=red](AF)(FF)
\Edge[color=red](BF)(EF)

\Edge[color=green](AA)(AE)
\Edge[color=green, style=dashed](AB)(AF)
\Edge[color=black](AA)(AF)
\Edge[color=black](AB)(AE)
\Edge[color=green](BA)(BE)
\Edge[color=green, style=dashed](BB)(BF)
\Edge[color=black](BA)(BF)
\Edge[color=black](BB)(BE)

\Edge[color=green, style=dashed](EA)(EE)
\Edge[color=green](EB)(EF)
\Edge[color=black, style=dashed](EA)(EF)
\Edge[color=black, style=dashed](EB)(EE)
\Edge[color=green, style=dashed](FA)(FE)
\Edge[color=green](FB)(FF)
\Edge[color=black, style=dashed](FA)(FF)
\Edge[color=black, style=dashed](FB)(FE)

\end{tikzpicture}
\end{tabular}
\caption{Constructing the $2$-d Adinkra (right) from two smaller $1$-d Adinkras (left) as a product. Note that the dashings are all ``consistent'' with the smaller Adinkras, except for the right-moving edges on the upper-left ``boundary'' of the rectangle; these correspond to right-moving edges where the grading corresponding to the first Adinkra has height $1$.\label{fig:product}}
\end{center}
\end{figure}

This definition is intended to be a graph-theoretic version of the tensor product construction in $\ZZ_2$-graded representations (see Appendix~\ref{app:repn} for more details).  The edges that come from $E_1$ give rise to left-moving edges, and the edges that come from $E_2$ give rise to right-moving edges.  It follows easily that for every vertex $v$ in $A_1\times A_2$ and for every color in $[n]$ there is a unique edge in $A_1\times A_2$ incident to $v$.  The fact that two-colored simple cycles have length $4$ follows from following cases, depending on whether the colors are both left-moving, both right-moving, or one of each.  The parity condition for an Adinkra also follows from considering these cases.  The properties related to the bigrading are straightforward.  We then have:

\begin{prop}
\label{prop:product-admissible}
Given Adinkras $A_1$ and $A_2$, $A_1\times A_2$ is a $2$-d Adinkra.
\end{prop}

\begin{definition}[Extending codes]
Let $p$ and $q$ be non-negative integers and let $n=p+q$.  Define $Z_L:\ZZ_2^p\to\ZZ_2^n$ to be the function that appends $q$ zeros, so that for instance, if $p=4$ and $q=3$, then $Z_L(1011)=1011000$.  Likewise, define $Z_R:\ZZ_2^q\to\ZZ_2^n$ to be the function that prepends $p$ zeros.

Our most common use of this notation is as follows: if $C$ is a binary block code of length $p$, we write $Z_L(C)$ for the image under $Z_L$.  Likewise, if $C$ is a binary block code of length $q$, we write $Z_R(C)$ for the image under $Z_R$.
\end{definition}

\begin{prop}
\label{prop:prodcode}
Let $A_1$ and $A_2$ be as above.  Then
\[C(A_1\times A_2)=Z_L(C(A_1))\oplus Z_R(C(A_2)).\]
\end{prop}
\begin{proof}
Let $(v_1,v_2)\in A_1\times A_2$.  Let $\vec{x}\in \ZZ_2^N$.  We can write $\vec{x}=\vec{x}_L+\vec{x}_R$ where $\vec{x}_L$ is zero in the last $q$ bits and $\vec{x}_R$ is zero in the first $p$ bits.  Now
\[\vec{x}(v_1,v_2)=(\vec{x}_L+\vec{x}_R)(v_1,v_2)=(\vec{x}_Lv_1,\vec{x}_Rv_2).\]
This means that $\vec{x}(v_1,v_2)=(v_1,v_2)$ if and only if $\vec{x}_Lv_1=v_1$ and $\vec{x}_R v_2=v_2$. So $\vec{x}\in C(A_1\times A_2)$ if and only if $\vec{x}_L\in Z_L(C(A_1))$ and $\vec{x}_R\in Z_R(C(A_2))$.
\end{proof}

\section{Codes for $2$-d Adinkras}
\label{sec:code2d}
Let $A$ be a connected $2$-d Adinkra with $(p,q)$ colors.  Then there is a doubly even code $C(A)$ associated with $A$.  But as a $2$-d Adinkra, we make a distinction between the first $p$ colors and the last $q$ colors, which for a code translates to the first $p$ bits and the last $q$ bits. A natural question is: ``knowing that an $1$-d Adinkra can be enriched into a $2$-d Adinkra, what else can we say about its code?'' In this section, we address this question.

\begin{definition}[Weights for left-moving and right-moving colors; ESDE codes]
Recall that for any vector $\vec{x}\in\ZZ_2^n$, the \emph{weight} of $\vec{x}$, denoted $\wt(\vec{x})$, is the number of $1$'s in $\vec{x}$.  Likewise, $\wt_L(\vec{x})$ is the the number of $1$'s in the first $p$ bits and $\wt_R(\vec{x})$ is the number of $1$'s in the last $q$ bits of $\vec{x}$. Let a code $C$, along with the parameters $(p,q)$, be called a \emph{even-split doubly even (ESDE) code} if $C$ is doubly-even and all codewords $\vec{x}$ in $C$ have both $\wt_L(\vec{x})$ and $\wt_R(\vec{x})$ even.
\end{definition}

This definition of ESDE codes is due to \cite{hubsch:weaving}, which also proves the following: 
\begin{thm}
\label{thm:esde}
If $A$ is a connected $2$-d Adinkra with $(p,q)$ colors, then $C(A)$ is an ESDE.
\end{thm}

We now prove the converse of this theorem.  That is, given an ESDE code, there exist connected $2$-d Adinkra with that code.  This procedure is analogous to the Valise Adinkras in $1$-d,\cite{d2l:first,d2l:graph-theoretic} in that the possible values of each component $(h_L,h_R)$ of the bigrading is as small as possible, i.e., two values.

\begin{construction}
\label{cons:valise}
Let $C$ be an ESDE code.  We will describe a construction that provides a $2$-d Adinkra with code $C$, called the {\em Valise 2-d Adinkra}. First, since $C$ is doubly-even, there exists a connected $1$-d Adinkra $A$ with code $C(A) = C$ by Theorem~\ref{thm:1d-quotients}. Fix a vertex $\overline{0}$ of $A$.  Now for every vertex $v$ there is a vector $\vec{x}\in\ZZ_2^n$ so that $\vec{x}\overline{0}=v$.  Then define $h_L(v)=\wt_L(\vec{x})\pmod{2}$ and $h_R(v)=\wt_R(\vec{x})\pmod{2}$.  Note that these functions are well-defined since $C$ is ESDE.  Then $(h_L, h_R)$ is a bigrading for $A$, making it a $2$-d Adinkra.  An example of the kind of $2$-d Adinkra that arises from this construction is Figure~\ref{fig:2d-example}.
\end{construction}

We therefore have:
\begin{thm}
\label{thm:esdeiff}
For a code $C \subset \ZZ_2^n$, there exists a $2$-d Adinkra $A$ with $C(A) = C$ if and only if $C$ is a ESDE code.
\end{thm}

The structure of the ESDE relates to interesting features of the colored graph of $A$.  Let $A_L$ be the $1$-d Adinkra with $p$ colors that consists of only the left-moving edges of $A$.  Let $A_R$ be the $1$-d Adinkra with $q$ colors that consists of only the right-moving edges of $A$ (where we shift the colors so that they range from $1$ to $q$ instead of $p+1$ to $p+q$).  Pick a vertex $\overline{0}$ in $A$.  Let $A_L^0$ be the connected component of $A_L$ containing $\overline{0}$ and let $A_R^0$ be the connected component of $A_R$ containing $\overline{0}$.

We now see that the codes for $A_L^0$ and $A_R^0$ (with an appropriate number of $0$s added to the left or right as necessary) provide important linear subspaces of $C(A)$.

\begin{prop}
\[Z_L(C(A_L^0))=C(A)\cap Z_L(\ZZ_2^p)\]
\[Z_R(C(A_R^0))=C(A)\cap Z_R(\ZZ_2^q)\]
\end{prop}
In other words, the codewords that are zero in the last $q$ bits are precisely the codewords from $C(A_L^0)$ with $q$ zeros appended to the right; and the codewords that are zero in the first $p$ bits are precisely the codewords from $C(A_R^0)$ with $p$ zeros prepended to the left.
\begin{proof}
If $\vec{x}\in Z_L(C(A_L^0))$, so that $\vec{x}=Z_L(\vec{y})$ for some $\vec{y}\in C(A_L^0)$.  Then trivially $\vec{x}\in Z_L(\ZZ_2^p)$.  Furthermore, $\vec{y}\overline{0}=\overline{0}$ in $A_L^0$.  In $A$, $Z_L(\vec{y})=\vec{x}$ does exactly the same thing, so $\vec{x}\overline{0}=\overline{0}$ in $A$.  Therefore $\vec{x}\in C(A)$.

Conversely if $\vec{x}\in C(A)\cap Z_L(\ZZ_2^p)$, then by the definition of the group action there is a path in $A$ from $\overline{0}$ to $\overline{0}$ following the colors corresponding to the $1$s in $\vec{x}$.  Since $\vec{x}\in Z_L(\ZZ_2^p)$, we have that this path only consists of left-moving colors, and so lies in $A_L^0$.

The proof for $Z_R(C(A_R^0))$ is similar.
\end{proof}

\begin{cor}
\label{cor:cplus}
\[Z_L(C(A_L^0))\oplus Z_R(C(A_R^0))\subset C(A)\]
\end{cor}
Comparing with Proposition~\ref{prop:prodcode}, this corollary says that the code for $A_L^0\times A_R^0$ is a linear subspace of the code for $A$.

Simply for brevity (and thus, readability) in later descriptions, we define using the product construction above $A'=A_L^0\times A_R^0$, and the code $C\,'=Z_L(C(A_L^0))\oplus Z_R(C(A_R^0))$.  Also, we write $C$ for $C(A)$.  Then $C(A')=C\,'$ and $C\,'\subset C$.

\begin{lem}
\label{lem:existk}
There exists a binary linear block code $K$ so that
\[C=C\,' \oplus K.
\]
\end{lem}
\begin{proof}
From Corollary~\ref{cor:cplus} and basic linear algebra, there exists a vector subspace $K$ of $\ZZ_2^n$ that is a vector space complement of
$C\,'$ in $C$.
\end{proof}
Note that $K$ is not necessarily uniquely defined.  It is, however, uniquely defined up to adding vectors in $C\,'$.  So a more invariant approach would be to use $C/C\,'$ instead of $K$, but $K$ has the advantage of being a code, therefore more concrete for computational purposes.

The interpretation of $K$ can be obtained by examining the set $V^0=A_L^0\cap A_R^0$.  Since $A_L^0$ has only left-moving edges and $A_R^0$ has only right-moving edges, $V^0$ has no edges at all: only vertices.  Furthermore, for every $v\in V^0$, we have $h_L(v)=h_L(\overline{0})$ and $h_R(v)=h_R(\overline{0})$ so all of the vertices in $V^0$ have the same bigrading.

We now show that there is a bijection between $K$ and $V^0$. As before, for every $\vec{x}\in\ZZ_2^n$, we write $\vec{x}=\vec{x}_L+\vec{x}_R$, where $\vec{x}_L$ is zero in the last $q$ bits and $\vec{x}_R$ is zero in the first $p$ bits.  Using this notation, we have the following theorem:

\begin{thm}
\label{thm:kv0}
The map $\Psi:K\to V^0$
given by $\Psi(\vec{x})=\vec{x}_L\overline{0}$ is a bijection.
\end{thm}
\begin{proof}
If $\vec{x}\in K\subset C(A)$, then $(\vec{x}_L+\vec{x}_R)\overline{0}=\overline{0}$, so $\vec{x}_L\overline{0}=\vec{x}_R\overline{0}$.  So $\Psi(\vec{x})\in A_L^0\cap A_R^0=V^0$.

To prove $\Psi$ is one-to-one, suppose $\Psi(\vec{x})=\Psi(\vec{y})$.  Then $\vec{x}_L\overline{0}=\vec{y}_L\overline{0}$.  Therefore $\vec{x}_L+\vec{y}_L\in C(A_L^0)$.  Likewise $\vec{x}_R+\vec{y}_R\in C(A_R^0)$.  So $\vec{x}+\vec{y}\in C\,'$.  Since $C\,'\cap K=\{\vec{0}\}$, we have that $\Psi$ is one-to-one.

To prove $\Psi$ is onto, let $v\in V^0$.  Since $v\in A_L^0$, there exists $\vec{x}_L$ with last $q$ bits zero, so that $\vec{x}_L\overline{0}=v$ in $A_L^0$.  Likewise there exists $\vec{x}_R$ with first $p$ bits zero, so that $\vec{x}_R\overline{0}=v$ in $A_R^0$.  Then $\vec{x}=\vec{x}_L+\vec{x}_R\in C(A)$.  Since $C(A)=C\,'\oplus K$, we can write $\vec{x}=\vec{c}+\vec{k}$ where $\vec{c}\in C\,'$ and $\vec{k}\in K$.  By the fact that $C\,'=Z_L(C(A_L^0))\oplus Z_R(C(A_R^0))$, we have that $\vec{c}_L\in Z_L(C(A_L^0))$ and so $\vec{c}_L\overline{0}=\overline{0}$.  Then
$\Psi(\vec{k})=\vec{k}_L\overline{0}
=(\vec{x}_L-\vec{c}_L)\overline{0}
=\vec{x}_L(\vec{c}_L(\overline{0}))
=\vec{x}_L(\overline{0})
=v.$
\end{proof}

\begin{ex}
Let $p=4$ and $q=2$ and the generating matrix for $C$ be
\[\left[\begin{array}{cccc|cc}
1&1&1&1&0&0\\
0&0&1&1&1&1
\end{array}\right].\]
Then $Z_L(C(A_L^0))$ has generating vector $\vec{m} = \left[\begin{array}{cccc|cc}
1&1&1&1&0&0\\
\end{array}\right]$ and $Z_R(C(A_R^0))=\{\vec{0}\}$, the trivial code.  We therefore see that $A_L^0$ is a $1$-d Adinkra with four colors with code with generating vector
$\left[\begin{array}{cccc}
1&1&1&1\\
\end{array}\right]$ and $A_R^0$ is a $1$-d Adinkra with two colors with trivial code.  The code $K$ can be chosen to be generated by $\vec{k} = \left[\begin{array}{cccc|cc}
0&0&1&1&1&1
\end{array}\right]$. Another choice for $K$ would have been the code generated by 
\[\vec{k} + \vec{m} = \left[\begin{array}{cccc|cc}
1&1&0&0&1&1
\end{array}\right].\]

See Figure~\ref{fig:example-quotient} for this example. To use Theorem~\ref{thm:kv0}, we start at $\overline{0}$ and follow an edge of color $3$, then an edge of color $4$, which uses the left-moving edges in $001111$.  This brings us to a new vertex, which is in $V^0$.  This vertex, and $\overline{0}$ itself, are the two elements of $V^0$, corresponding to the two elements of $K$. 
\end{ex}

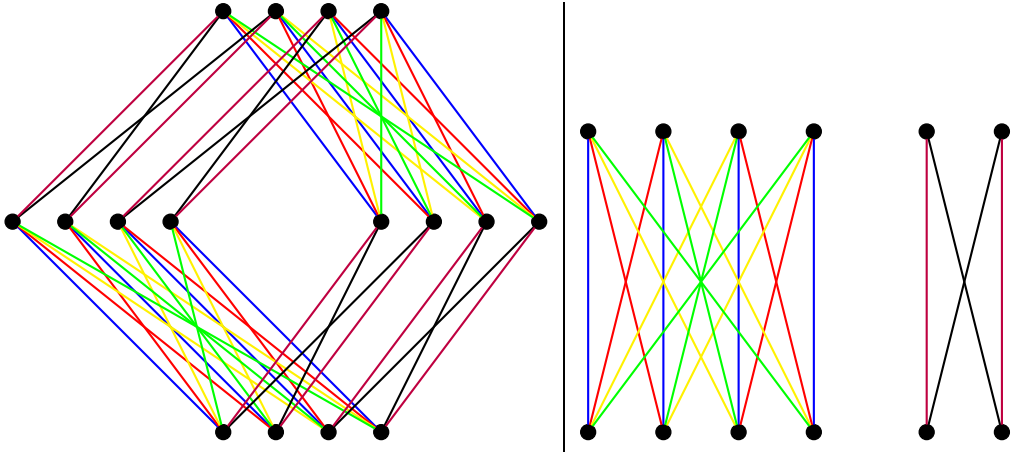
\begin{figure}
\begin{center}
\begin{tabular}{c|c}
\begin{tikzpicture}[scale=0.07]
\SetVertexSimple[MinSize=5pt]
\SetUpEdge[labelstyle={draw}]
\Vertex[x=0,y=0]{000000}
\Vertex[x=-40,y=40]{000100}
\Vertex[x=-30,y=40]{001000}
\Vertex[x=10,y=0]{001100}
\Vertex[x=-20,y=40]{010000}
\Vertex[x=20,y=0]{010100}
\Vertex[x=30,y=0]{011000}
\Vertex[x=-10,y=40]{011100}
\Vertex[x=30,y=40]{000001}
\Vertex[x=0,y=80]{000101}
\Vertex[x=10,y=80]{001001}
\Vertex[x=40,y=40]{001101}
\Vertex[x=20,y=80]{010001}
\Vertex[x=50,y=40]{010101}
\Vertex[x=60,y=40]{011001}
\Vertex[x=30,y=80]{011101}

\Edge[color=blue](000000)(000100)
\Edge[color=blue](001000)(001100)
\Edge[color=blue](010000)(010100)
\Edge[color=blue](011000)(011100)
\Edge[color=blue](000001)(000101)
\Edge[color=blue](001001)(001101)
\Edge[color=blue](010001)(010101)
\Edge[color=blue](011001)(011101)

\Edge[color=red](000000)(001000)
\Edge[color=red](000100)(001100)
\Edge[color=red](010000)(011000)
\Edge[color=red](010100)(011100)
\Edge[color=red](000001)(001001)
\Edge[color=red](000101)(001101)
\Edge[color=red](010001)(011001)
\Edge[color=red](010101)(011101)

\Edge[color=yellow](000000)(010000)
\Edge[color=yellow](011100)(001100)
\Edge[color=yellow](000100)(010100)
\Edge[color=yellow](011000)(001000)
\Edge[color=yellow](000001)(010001)
\Edge[color=yellow](001001)(011001)
\Edge[color=yellow](000101)(010101)
\Edge[color=yellow](001101)(011101)

\Edge[color=green](000000)(011100)
\Edge[color=green](001000)(010100)
\Edge[color=green](010000)(001100)
\Edge[color=green](011000)(000100)
\Edge[color=green](000001)(011101)
\Edge[color=green](001001)(010101)
\Edge[color=green](010001)(001101)
\Edge[color=green](011001)(000101)

\Edge[color=black](000000)(001101)
\Edge[color=black](000100)(001001)
\Edge[color=black](010000)(011101)
\Edge[color=black](010100)(011001)
\Edge[color=black](000001)(001100)
\Edge[color=black](000101)(001000)
\Edge[color=black](010001)(011100)
\Edge[color=black](010101)(011000)

\Edge[color=purple](000000)(000001)
\Edge[color=purple](000100)(000101)
\Edge[color=purple](010000)(010001)
\Edge[color=purple](010100)(010101)
\Edge[color=purple](001000)(001001)
\Edge[color=purple](001100)(001101)
\Edge[color=purple](011000)(011001)
\Edge[color=purple](011100)(011101)
\end{tikzpicture} &
\begin{tikzpicture}[scale=0.1]
\SetVertexSimple[MinSize=5pt]
\SetUpEdge[labelstyle={draw}]
\Vertex[x=0,y=0]{000000}
\Vertex[x=0,y=40]{000100}
\Vertex[x=10,y=40]{001000}
\Vertex[x=10,y=0]{001100}
\Vertex[x=20,y=40]{010000}
\Vertex[x=20,y=0]{010100}
\Vertex[x=30,y=0]{011000}
\Vertex[x=30,y=40]{011100}

\Edge[color=blue](000000)(000100)
\Edge[color=blue](001000)(001100)
\Edge[color=blue](010000)(010100)
\Edge[color=blue](011000)(011100)
\Edge[color=red](000000)(001000)
\Edge[color=red](000100)(001100)
\Edge[color=red](010000)(011000)
\Edge[color=red](010100)(011100)
\Edge[color=yellow](000000)(010000)
\Edge[color=yellow](011100)(001100)
\Edge[color=yellow](000100)(010100)
\Edge[color=yellow](011000)(001000)
\Edge[color=green](000000)(011100)
\Edge[color=green](001000)(010100)
\Edge[color=green](010000)(001100)
\Edge[color=green](011000)(000100)

\Vertex[x=45,y=0]{000000'}
\Vertex[x=55,y=0]{001100'}
\Vertex[x=45,y=40]{000001'}
\Vertex[x=55,y=40]{001101'}
\Edge[color=black](000000')(001101')
\Edge[color=black](000001')(001100')
\Edge[color=purple](000000')(000001')
\Edge[color=purple](001100')(001101')
\end{tikzpicture}
\end{tabular}
\end{center}
\caption{Left: a $2$-d Adinkra with dashes omitted. Right: picking any vertex and looking at connected components give us a pair of $1$-d Adinkras $A_L^0$ and $A_R^0$, the product of which contains as a quotient the original Adinkra; in this case we quotient via a $1$-dimensional code $K$ of size $2^1 = 2$, obtaining the desired $8 \times 4 / 2 = 16$ vertices.  \label{fig:example-quotient}}
\end{figure}

\begin{ex}
Let $p=q=4$ and consider the following generating matrix for $C$:
\[\left[\begin{array}{cccc|cccc}
1&1&0&0&1&1&0&0\\
1&1&1&1&1&1&1&1\\
0&0&1&1&1&1&0&0\\
1&0&1&0&1&0&1&0
\end{array}\right]\]
If we let $\vec{x}_1, \ldots, \vec{x}_4$ be the rows of this matrix, we see that $\vec{x}_1+\vec{x}_3$ has zeros on the right side of the vertical line.  Other than the zero word, no other combination has all zeros on the right side, so $Z_L(C(A_L^0))$ has generating matrix/vector $\left[\begin{array}{cccc|cccc}
1&1&1&1&0&0&0&0
\end{array}\right]$ and $C(A_L^0)$ has generating matrix/vector $\left[\begin{array}{cccc}
1&1&1&1
\end{array}\right].$

Likewise we can find $\vec{x}_1+\vec{x}_2+\vec{x}_3$ which is the unique nonzero codeword with all zeros on the left side, and so $Z_R(C(A_R^0))$ has generating matrix/vector
\[\left[\begin{array}{cccc|cccc}
0&0&0&0&1&1&1&1
\end{array}\right]\]
and $C(A_R^0)$ has generating matrix $\left[\begin{array}{cccc}
1&1&1&1
\end{array}\right].$ Then $A_L^0$ and $A_R^0$ are both $1$-d Adinkras with $4$ colors with code generated by $1111$, and $K$ can be taken to be (for instance)
\[\left[\begin{array}{cccc|cccc}
0&0&1&1&1&1&0&0\\
1&0&1&0&1&0&1&0
\end{array}\right].\]
Standard arguments in linear algebra allow us to choose the generating basis for $C$ to consist of a generating basis for $Z_L(C(A_L^0))$, then a generating basis for $Z_R(C(A_R^0))$, then a generating basis for $K$.  In this case, we would write
\[\left[\begin{array}{cccc|cccc}
1&1&1&1&0&0&0&0\\\hline
0&0&0&0&1&1&1&1\\\hline
0&0&1&1&1&1&0&0\\
1&0&1&0&1&0&1&0
\end{array}\right]\]
where the horizontal lines separate the three subspaces.  Each line has weight a multiple of $4$, where the first two lines have the $1$s all on one side or the other, while the last two lines (the ones responsible for $K$) have the $1$s split on both sides in a way that both sides have even weight.

Then $V^0$ has four elements: $\overline{0}$, $(00110000)\overline{0}$, $(10100000)\overline{0}$, and $(10010000)\overline{0}$.
\end{ex}

While not necessary for proving the main theorem of this paper, \cite{hubsch:weaving} also asked how to classify ESDE codes. It turns out that the answer is fairly concise. To do this, it is useful to extend the notion of the splitting of $n$ into $p+q$.  In particular, instead of insisting that the left-moving colors be written as the first $p$ bits, we partition $[n]$ into $[n]=L\cup R$ such that $|L|=p$ and $|R|=q$.  We fix $n$ and a doubly even code $C$, then characterize which partitions into $L$ and $R$ make $C$ an ESDE.

\begin{thm}
\label{thm:esdeclassify}
Given a doubly even code $C$ of length $n$, there is a bijection between codewords in $C^\perp$ and (ordered) partitions $L \cup R$ that make $C$ into an ESDE. There are $2^{n-k}$ such partitions. 
\end{thm}
\begin{proof}
Consider a partition $L \cup R$ that makes $C$ into an ESDE, and consider the codeword $w$ that is defined to have $1$ at all the positions in $L$ and $0$ at all the positions in $R$.  By definition of ESDE codes, all codewords in $C$ have an even number of $1$'s in the support of $w$, which is equivalent to saying that $w$ is orthogonal to all the codewords in $C$. Thus, $w \in C^\perp$.  Conversely, for any $w \in C^\perp$, $w$ is orthogonal to all codewords in $C$ and thus give an ESDE. Thus, there is a bijection between the two sets.  Note these are ordered partitions; the codeword which has $1$ at all the positions in $R$ and $0$ otherwise would give the same partition, but in reversed order.
\end{proof}

\section{Proof of Main Theorem}
\label{sec:quotient}
In this section we prove the main theorem of the paper, Theorem~\ref{thm:main}.  This refers to a connected $2$-d Adinkra $A$ with $(p,q)$ colors.  As in Section~\ref{sec:code2d} we pick a vertex $\overline{0}$ in $A$ and define $A_L^0$ and $A_R^0$.  We use Construction~\ref{const:product} to define $A'=A_L^0\times A_R^0$ and let $C=C(A)$ and $C\,'=C(A')=Z_L(C(A_L^0))\oplus Z_R(C(A_R^0))$.  Let $K$ be a code given by Lemma~\ref{lem:existk}.

We now try to prove the following (slight) restatement of Theorem~\ref{thm:main}:
\begin{thm}
\label{thm:quotient}
Let $A$ be a connected $2$-d Adinkra.  Then there is a vertex switching $F$ and an action of $K$ on $F(A')$ that preserves colors, dashing, and bigrading, and so that
\[F(A')/K\cong A\]
as an isomorphism of $2$-d Adinkras (that is, as an isomorphism of graphs that preserves colors, dashing, and bigrading).
\end{thm}

We prove Theorem~\ref{thm:quotient} in two steps, by first constructing a (color and bigrading preserving) graph isomorphism $\tilde{\Phi}:A'/K\to A$ and then finding a suitable vertex switching $F$.

\begin{thm}
\label{thm:isocolors}
The code $K$ acts on $A'$ via color preserving isomorphisms to produce a quotient $A'/K$.  There is a color preserving graph epimorphism $\Phi:A' \to A$ that sends $(\overline{0},\overline{0})$ to $\overline{0}$.  This descends to $A'/K$ to produce a color preserving graph isomorphism $\tilde{\Phi}:A'/K\to A$.
\end{thm}
\begin{proof}
Let $I^n$ be the colored Hamming cube: that is, a graph with vertex set $\{0,1\}^n$ and two vertices are connected with an edge of color $i$ if they differ only in bit $i$.  Recall from Theorem~\ref{thm:1d-quotients} that every connected Adinkra is, as a colored graph, the quotient of $I^n$ by the code for the Adinkra.  So we have $I^n/C \cong A$ and $I^n/C\,'\cong A'$. These are isomorphisms as colored graphs.  They can be chosen so that $\vec{0}=(0,\ldots,0)$ is sent to $\overline{0}$ in $A$ and $(\overline{0},\overline{0})$ in $A'$, respectively.

Now $K$ is a doubly even code, and $K\cap C\,'=0$, so $K$ acts on $I^n/C\,'$ and on $A'$ in a way that nontrivial elements of $K$ move vertices a distance greater than $2$. By the content of the proof of the extension of Theorem~\ref{thm:1d-quotients} in \cite{zhang:adinkras}, this means we can quotient the colored graph $I^n/C\,'$, and thus $A'$, by $K$.  We then have the following commutative diagram of colored graphs:

\[
\begin{CD}
I^n/C\,' @>i_1>\cong> A'\\
@VV\pi_1 V @VV\pi_2 V\\
(I^n/C\,')/K @>i_2>\cong> A'/K\\
\end{CD}  
\]

A standard argument gives
\[(I^n/C\,'\,)/K \cong I^n/(C\,'\oplus K)=I^n/C,\]
and adding this to the above commutative diagram, we then have:
\[
\xymatrixcolsep{5pc}
\xymatrix{
I^n/C\,' \ar[d]^{\pi_1} \ar[r]^{i_1}_{\cong} &A'\ar[d]^{\pi_2}\ar@/^2pc/[dd]^{\Phi} \\
(I^n/C\,')/K \ar[d]^{i_4}_{\cong} \ar[r]^{i_2}_\cong &A'/K\ar[d]_{\cong}^{\tilde{\Phi}}\\
I^n/(C\,'\oplus K) \ar[r]^{i_3}_\cong &A.}
\]
where $\tilde{\Phi}=i_3\circ i_4\circ i_2^{-1}$ and $\Phi=\tilde{\Phi}\circ \pi_2$.  Then $\tilde{\Phi}$ is an isomorphism of colored graphs, and $\Phi$ is an epimorphism of colored graphs.  Standard diagram chasing shows that $\Phi(\overline{0},\overline{0})=\overline{0}$.
\end{proof}

It will be useful to have the following result:
\begin{lem}
\label{lem:gphi}
If $\vec{x}\in \ZZ_2^n$, then $\Phi(\vec{x}(v_1,v_2))=\vec{x}\Phi(v_1,v_2).$
\end{lem}
\begin{proof}
For $\vec{x}=\vec{e}_i$, the vector that is $1$ in component $i$ and $0$ otherwise, this lemma is the statement that $\Phi$ is color-preserving.  By composing many maps of this type, we get the statement for all vectors $\vec{x}\in \ZZ_2^n$.
\end{proof}

\begin{lem}
\label{lem:phisides}
For all $v\in A_L^0$, $\Phi(v,\overline{0})=v$, and for all $w\in A_R^0$, $\Phi(\overline{0},w)=w$.  In particular, $\Phi$ restricted to $A_L^0\times\{\overline{0}\}$ is an isomorphism onto its image, and likewise for $\Phi$ restricted to $\{\overline{0}\}\times A_R^0$.
\end{lem}
\begin{proof}
Let $\vec{x}\in \ZZ_2^p$ be such that $\vec{x}\overline{0}=v$.  By Lemma~\ref{lem:gphi}, $\Phi(v,0) = \Phi(Z_L(\vec{x})(\overline{0},\overline{0})) = Z_L(\vec{x})\Phi(\overline{0},\overline{0})=Z_L(\vec{x})\overline{0}=v$.  The proof for $w$ is similar.
\end{proof}

\begin{lem}
\label{lem:phiformula}
Let $(v,w)\in A'$, with $w=\vec{x}\overline{0}$.  Then $\Phi(v,w)=Z_R(\vec{x})v.$
\end{lem}
\begin{proof}
We have $\Phi(v,\overline{0})=v$ from Lemma~\ref{lem:phisides}.  Act on both sides with $Z_R(\vec{x})$, and using Lemma~\ref{lem:gphi}, the result follows.
\end{proof}

\begin{lem}
\label{lem:mainepibigrading}
The graph epimorphism $\Phi$ preserves the bigrading.
\end{lem}
\begin{proof}
Let $(v,w)\in A'$ with $w=\vec{x}\overline{0}$.  By Lemma~\ref{lem:phiformula}, we get
\[h_L(\Phi(v,w))=h_L(Z_R(\vec{x})v).\]
Since $Z_R(\vec{x})$ follows right-moving colors, this does not affect $h_L$, and so the above is equal to $h_L(v)$.  In $A'=A_L^0\times A_R^0$, this is $h_L(v,w)$. Therefore $\Phi$ preserves $h_L$.  The fact that it preserves $h_R$ is proved similarly.
\end{proof}

\begin{lem}
\label{lem:kgrading}
If $\vec{x}\in K$, and $(v_1,v_2)\in A_L^0\times A_R^0$, then $\vec{x}(v_1,v_2)$ has the same bigrading as $(v_1,v_2)$.
\end{lem}
\begin{proof}
This follows from Lemma~\ref{lem:mainepibigrading} and the fact that $K\subset C(A)$, so if $\vec{x}\in K$, then $\vec{x}\Phi(v_1,v_2)=\Phi(v_1,v_2)$.
\end{proof}

\begin{thm}
\label{thm:isograding}
The code $K$ acts on $A'$ via color and bigrading preserving isomorphisms to produce a quotient $A'/K$.  The map $\Phi$ (resp. $\tilde{\Phi}$) is a color and bigrading preserving epimorphism (resp. isomorphism).
\end{thm}
\begin{proof}
This theorem builds on Theorem~\ref{thm:isocolors}.  Lemma~\ref{lem:kgrading} means that the action of $K$ on $A'$ preserves the bigrading.  Lemma~\ref{lem:mainepibigrading} provides the rest of this theorem.
\end{proof}

Unfortunately, it is too much to expect $\Phi$ to preserve the dashing, or even that the dashing on $A'=A_L^0\times A_R^0$ is invariant under the action of $K$ (so that $A'/K$ could have an obviously well-defined dashing).  However, if we allow the operation of \emph{vertex switching}, then we can basically accomplish these goals, giving $F$ of Theorem~\ref{thm:quotient}.

Consider the dashing $\mu$ on $A$.  This restricts to $A_L^0$ and $A_R^0$, and Construction~\ref{const:product} produces a dashing $\mu_1$ on $A'=A_L^0\times A_R^0$.   The graph homomorphism $\Phi:A'\to A$ pulls back the dashing $\mu$ to $\mu_2$ on $A'$. While $\mu_1$ and $\mu_2$ can be different, they agree on the following parts of the Adinkra:

\begin{lem}
\label{lem:agree-on-boundary}
The dashings $\mu_1$ and $\mu_2$ agree on $A_L^0\times \{\overline{0}\}$ and on $\{\overline{0}\}\times A_R^0$.
\end{lem}
\begin{proof}
The construction of $\mu_1$ gives each edge in $A_L^0\times\{\overline{0}\}$ the same dashing as in $A_L^0$ under the association of every edge $(v,w)$ with $((v,\overline{0}),(w,\overline{0}))$.  Lemma~\ref{lem:phisides} shows that the same is true for $\mu_2$.  Therefore $\mu_1$ and $\mu_2$ agree on $A_L^0\times\{\overline{0}\}$.  Likewise for $\{0\}\times A_R^0$.
\end{proof}

\begin{lem}
\label{lem:cycles-switching-class}
Two dashings $\mu$ and $\bar{\mu}$ on an Adinkra $A$ have the same parity on all cycles\footnote{This type of result has a natural reformulation with homological algebra, done in independent ways by the first author's work using cubical cohomology \cite{dil:cohomology} and the second author's work using CW-complexes \cite{zhang:adinkras}. In either formulation, having parities of $\mu$ and $\bar{\mu}$ agree on cycles is equivalent to $\mu - \bar{\mu_2} = 0$ in cohomology.} if and only if there is a vertex switching on $A$ that turns $\mu$ into $\bar{\mu}$.
\end{lem}

\begin{proof}
Since vertex switching preserves parity on any cycle, the ``if'' direction is trivial and it suffices to prove the other direction.

Assume $\mu$ and $\bar{\mu}$ have the same parity on all cycles. It suffices to prove the statement for $A$ connected, since we can repeat our argument on each connected component of $A$.

Next, we shall prove that for any tree $T$ that is a subgraph of $A$, there exists a vertex switching $F$ on $A$ so that $F(\mu)$ and $\bar{\mu}$ agree on $T$.  This can be proved by induction the number of vertices in $T$.  The base case of one vertex is trivial.  If $T$ has more than one vertex, then there is a leaf $v$ in $T$ incident to only one edge $e$ in $T$.  Let $T_0$ be the tree with the vertex $v$ and edge $e$ omitted.  Then $T_0$ has one fewer vertex than $T$ so by the inductive hypothesis, there is a vertex switching $F_0$ on $A$ so that $F_0(\mu)$ and $\bar{\mu}$ agree on $T_0$.  If $F_0(\mu)(e)\not=\bar{\mu}(e)$, then let $F$ be the vertex switching $F_0$ followed by a vertex switching at $v$; otherwise let $F=F_0$.  Then $F(\mu)$ and $F_0(\mu)$ agree on $T_0$, so $F(\mu)$ agrees with $\bar{\mu}$ on all of $T$.

Now in the case where $T$ is a spanning tree (so that it is maximal), we claim that $F(\mu)$ and $\bar{\mu}$ agree on all of $A$.  Consider any edge $e$ not in $T$. This edge completes at least one cycle with edges in $T$ (otherwise $T$ was not a spanning tree). Since the two cycles have the same parity in $F(\mu)$ and $\bar{\mu}$ by assumption, and the $F(\mu)$ and $\bar{\mu}$ agree on all edges in the cycle except for $e$, they must agree on $e$ as well.  Thus, $F(\mu)=\bar{\mu}$ on all of $A$.
\end{proof}

Now we return to the two dashings $\mu_1$ and $\mu_2$ on $A'$.  Based on what we just proved, the following lemma will assure the existence of a vertex switching that sends $\mu_1$ to $\mu_2$.  It uses the fact that these two dashings agree on $A_L^0\times\{\overline{0}\}$ and $\{\overline{0}\}\times A_R^0$ (Lemma~\ref{lem:agree-on-boundary}).  In terms of the cubical cohomology, this result is a kind of K\"unneth theorem.

\begin{lem}
\label{lem:switch12}
The parities of $\mu_1$ and $\mu_2$ agree on all cycles of $A'$.
\end{lem}
\begin{proof}
Let our cycle be $(v_0,v_1,\ldots,v_k)$ with $v_0=v_k$.  We first consider the case where $v_0=\overline{0}$.

For this proof, we define a \emph{color sequence} of a path to be the sequence of colors $(c(v_0,v_1),c(v_1,v_2),\ldots,c(v_{k-1},v_k))$ of edges along the path.  Note that given a starting vertex $v_0$ and a color sequence, there is a unique path that starts at $v_0$ with that color sequence\footnote{Recall in Section~\ref{sec:code} we treated this sequence of colors as a $\ZZ_2^n$ action on the underlying \textbf{graph}, where the order did not matter. In this proof we are not just traversing the graph but also keeping track of the sign of the dashings, so we have to keep in mind the order.}.  This follows by applying induction to Property $2$ of the definition of an Adinkra. 

We begin with the color sequence for the cycle $(v_0,\ldots,v_k)$.  We will now describe a series of modifications to this cycle, described by modifying the color sequence.  The idea is to perform a ``bubble sort'', by iteratively swapping adjacent colors until the left-moving colors are all at the beginning and the right-moving colors are all at the end.

First, given a color sequence
\[(c_1,\ldots,c_{j-1},c_j,c_{j+1},c_{j+2},\ldots,c_k), \]
an adjacent swap results in a color sequence
\[(c_1,\ldots,c_{j-1},c_{j+1},c_j,c_{j+2},\ldots,c_k).\]
Modifying a color sequence in this way leads to a new path from $\overline{0}$.  The path is unchanged up to $v_{j-1}$, but by the definition of Adinkras, property 2, the path returns to $v_{j+1}$ so it is only $v_j$ that has changed (see Figure~\ref{fig:colorswap}).  Thus, the new path is still a cycle starting at $\overline{0}$.  The effect on the parity of any dashing is, by property 3, to add $1$ modulo $2$.  In particular, $\mu_1$ and $\mu_2$ are both affected in the same way.

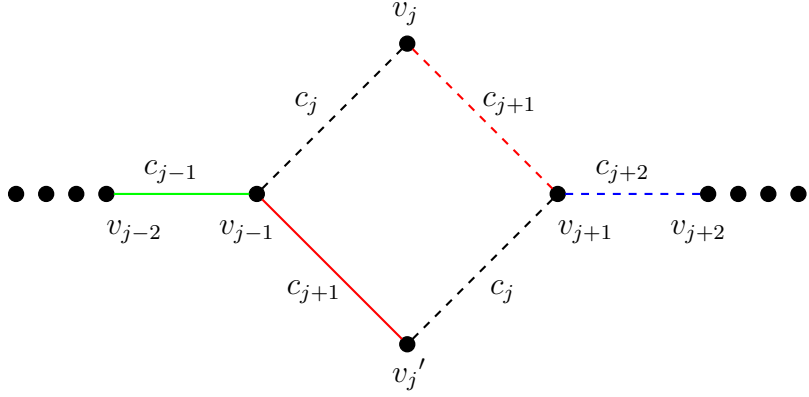
\begin{figure}
\begin{center}
\begin{tikzpicture}[scale=0.1]
\SetVertexSimple[MinSize=5pt]
\SetUpEdge[labelstyle={draw}]
\Vertex[x=0,y=0]{A}
\Vertex[x=20,y=0]{B}
\Vertex[x=40,y=20]{C}
\Vertex[x=40,y=-20]{D}
\Vertex[x=60,y=0]{E}
\Vertex[x=80,y=0]{F}
\Edge[color=green](A)(B)
\Edge[color=black,style=dashed](B)(C)
\Edge[color=red,style=dashed](C)(E)
\Edge[color=red](B)(D)
\Edge[color=black,style=dashed](D)(E)
\Edge[color=blue,style=dashed](E)(F)
\Vertex[x=-4,y=0]{E1}
\Vertex[x=-8,y=0]{E2}
\Vertex[x=-12,y=0]{E3}
\Vertex[x=84,y=0]{E4}
\Vertex[x=88,y=0]{E5}
\Vertex[x=92,y=0]{E6}
\node[text width=3cm] at (15, -5) {$v_{j-2}$};
\node[text width=3cm] at (30, -5) {$v_{j-1}$};
\node[text width=3cm] at (53, 24) {$v_{j}$};
\node[text width=3cm] at (53, -24) {$v_{j}{}'$};
\node[text width=3cm] at (75, -5) {$v_{j+1}$};
\node[text width=3cm] at (90, -5) {$v_{j+2}$};
\node[text width=3cm] at (20,3) {$c_{j-1}$};
\node[text width=3cm] at (40,12) {$c_{j}$};
\node[text width=3cm] at (65,12) {$c_{j+1}$};
\node[text width=3cm] at (39,-13) {$c_{j+1}$};
\node[text width=3cm] at (66,-13) {$c_{j}$};
\node[text width=3cm] at (80,3) {$c_{j+2}$};
\end{tikzpicture}
\end{center}
\caption{An adjacent swap of colors.  If we swap colors in position $j$ and $j+1$, vertex $v_j$ will be replaced by vertex $v_j{}'$, and the path gets modified from the upper path to the lower path in the diagram.  The parity of the path changes by exactly $1$ modulo $2$ because the square above has odd parity.
\label{fig:colorswap}}
\end{figure}

It is straightforward to find a series of adjacent swaps so that the left-moving colors are moved to the beginning of the color sequence.  Then the resulting path starts from $\overline{0}$, stays in $A_L^0\times\{0\}$, then follows right-moving edges, ending in $\overline{0}$.  Since the right-moving edges end in $\overline{0}$, it must be that the right-moving edges are in $\{0\}\times A_R^0$.  By Lemma~\ref{lem:agree-on-boundary}, $\mu_1$ and $\mu_2$ are equal here, and so their parities on this modified path are the same.  Therefore, their parities on the original loop were the same.

Now we consider loops $p$ where $v_0\not=\overline{0}$.  Since $A_L^0\times A_R^0$ is connected, there is a path $p_0$ from $\overline{0}$ to $v_0$.  Take the path $p_0$, followed by $p$, then followed by $p_0^{-1}$ (meaning $p_0$ traversed in the opposite sense).  This is a loop starting and ending in $\overline{0}$, but the parity of a dashing is the same as that of $p$, since every new edge in $p_0$ is counterbalanced by a new edge in $p_0^{-1}$.  Therefore the parities of $\mu_1$ and $\mu_2$ agree on all loops.
\end{proof}

We are now ready to put everything together and prove our main theorem.

\begin{proof}[Proof of Theorem~\ref{thm:quotient}]
Let $A$ be a connected $2$-d Adinkra, and define $A_L^0$, $A_R^0$, $A'$, $C\,'$, $K$, $\Phi$, and $\tilde{\Phi}$ as above.

Use Construction~\ref{const:product} to construct the $2$-d Adinkra $A'=A_L^0\times A_R^0$ with dashing $\mu_1$.
By Theorem~\ref{thm:isocolors} and Theorem~\ref{thm:isograding}, there is a graph homomorphism $\Phi:A' \to A$ that preserves colors and bigrading.  If we take the dashing $\mu$ from $A$ and pull it back using $\Phi$ to a dashing $\mu_2$ on $A'$, then Lemma~\ref{lem:switch12} and Lemma~\ref{lem:cycles-switching-class} together gives the existence of a vertex switching $F$ sending $\mu_1$ to $\mu_2$.

By Theorem~\ref{thm:isograding}, $K$ acts on $A'$ to produce $A'/K$, a well-defined $2$-d Adinkra without dashing, and an isomorphism $\tilde{\Phi}:A'/K\to A$ that preserves colors and the bigrading.  Since $\mu_2$ is invariant under $K$, we obtain  $F(A')/K$, a well-defined $2$-d Adinkra with dashing. Since $F(A')$ and $A'$ only differ in dashing, $\tilde{\Phi}$ is still an isomorphism that preserves colors and bigrading.  Since $\mu_2$ is obtained by pulling back $\mu$ from $A$, this isomorphism preserves dashing as well.
\end{proof}

\section{The Structure of $2$-d Adinkras}
\label{sec:structure}
Theorem \ref{thm:quotient} is very powerful; we immediately know a lot about what a $2$-d Adinkra must look like. Let the \emph{support} of a $2$-d Adinkra (and/or its bigrading function $(h_L,h_R)$) be defined as the range of $(h_L,h_R)$, its bigrading function. Then:
\begin{cor}
\label{cor:rectangle}
Let $A$ be a connected $2$-d Adinkra.  The support of $A$ is a rectangle.  That is, there exist integers $x_0$, $x_1$, $y_0$, and $y_1$ such that the support is
\[\{(i,j)\in\ZZ^2\,|\,x_0 \le i\le x_1\mbox{ and }y_0\le j\le y_1\}.\]
\end{cor}
\begin{proof}
Since $A_L^0$ is a connected $1$-d Adinkra and edges change $h_L$ by 1, there are integers $x_0$ and $x_1$ so that the range of the grading is $\{i\,|\,x_0\le i\le x_1\}$.  Likewise $A_R^0$ has a range of grading $\{j\,|\,y_0\le j\le y_1\}$ for some integers $y_0$ and $y_1$.

By Construction~\ref{const:product}, the support of $A'=A_L^0\times A_R^0$ is
\[\{(i,j)\,|\,x_0\le i\le x_1\mbox{ and }y_0\le j\le y_1\}.\]
Since $\Phi$ preserves the bigrading, this is the support of $A$ as well.
\end{proof}

\begin{prop}
\label{prop:componentsiso}
Let $A$ be a connected $2$-d Adinkra.  All connected components of $A_L$ (and respectively $A_R$) are isomorphic as graded posets.
\end{prop}
\begin{proof}
Consider a connected component of $A_L$.  Suppose $v$ is a vertex in this connected component.  Then there is a $\vec{x}\in\ZZ_2^n$ so that $\vec{x}\overline{0}=v$.  Then the map $f:A_L\to A_L$ with $f(w)=\vec{x}(w)$ is a color-preserving graph isomorphism.  Thus, it sends connected components onto connected components, and in particular, $A_L^0$ to the connected component containing $v$.
\end{proof}

From these results and the results of the previous section, we now fully know what a $2$-d Adinkra looks like. All of the graphical data of a $2$-d Adinkra is basically dictated by the connected components (one from left-moving colors and one from right-moving colors) of any single vertex and how they are glued together (this is what $K$ encodes); then these ``slices'' are put together into a rectangle.

\subsection{Constructing $2$-d Adinkras}
An alternate way to view Theorem~\ref{thm:quotient} is as a way to construct all $2$-d Adinkras.

\begin{construction}
First choose any doubly-even code and any codeword in its orthogonal complement. Theorem~\ref{thm:esdeclassify} shows this is exactly the amount of data we need to create an ESDC code $C$.  Write
\[C=C_L\oplus C_R\oplus K,\]
where $C_L=C\cap Z_L(\ZZ_2^p)$ and $C_R=C\cap Z_R(\ZZ_2^q)$. Use the quotient construction involved in Theorem~\ref{thm:1d-quotients} to create the $1$-d Adinkras $A_1=I^p/\pi_L(C_L)$ and $A_2=I^q/\pi_R(C_R)$, where $\pi_L:\ZZ_2^n\to\ZZ_2^p$ is projection onto the first $p$ bits and $\pi_R:\ZZ_2^n\to\ZZ_2^q$ is projection onto the last $q$ bits.

Second, we need a grading $h_1$ on $A_1$ that is invariant under $\pi_L(K)$ and a grading $h_2$ on $A_2$ that is invariant under $\pi_R(K)$. There are a finite number of rank functions to consider for each graph, so it is definitely possible in principle to generate all gradings\footnote{For specific algorithms, one can use either the ``hanging gardens'' construction in Ref.~\cite{d2l:graph-theoretic} or consider the vertices of the ``rank family poset'' from \cite{zhang:adinkras}.} though this is expected to be a large set for higher $n$. Construct the quotient $A=(A_1\times A_2)/K$ using Construction~\ref{const:product}.  This produces a graph with colors and a bigrading. The invariance of the gradings $h_1$ and $h_2$ makes this bi-grading well-defined under the quotient.

Finally, we put an admissible dashing on $A$. There is again a finite number of possible dashings, so this doable via an exhaustive process. We can obtain dashings on the $1$-d Adinkra $I^n/C$ and use $\Phi$ to pull them back to $A$. Recall the discussion after Theorem~\ref{thm:1d-quotients} for relevant results. 

\end{construction}

\begin{thm}
Every $2$-d Adinkra can obtained by this construction.
\end{thm}

\begin{proof}
Given any $2$-d Adinkra $A$, there is an ESDC code $C$.  Pick a vertex $\overline{0}$ and define $A_L^0$ and $A_R^0$ as in (\ref{thm:quotient}).

Restrict the gradings $h_L$ and $h_R$ onto $A_L^0$ and $A_R^0$.  Note that if $g\in C$, then $\pi_L(g)v=\pi_R(g)v$, and so $h_L(\pi_L(g)v)=h_L(v)$ and $h_R(\pi_L(g)v)=h_R(v)$.  Therefore $h_L$ restricted to $A_L^0$ is invariant under $\pi_L(K)$.  Likewise $h_R$ restricted to $A_R^0$ is invariant under $\pi_R(K)$. The dashings, as described in \cite{d2l:topology}, can be obtained by choosing the specific quotient $I^n/C$. Theorem~\ref{thm:quotient} gives a description of $A$ in terms of this construction.
\end{proof}

\begin{ex}
Consider the code given by the generating matrix/vector $\left[\begin{array}{cc|cc}
1&1&1&1
\end{array}\right].$ Then $C_L$ and $C_R$ are trivial, with $p=q=2$, and $K$ is generated by $\left[\begin{array}{cccc}
1&1&1&1
\end{array}\right].$ 

As graphs, the Adinkras $A_1$ and $A_2$ are both isomorphic to $I^2$. There are normally $2$ ways (up to relabeling of vertices) to put a rank function on $I^2$:

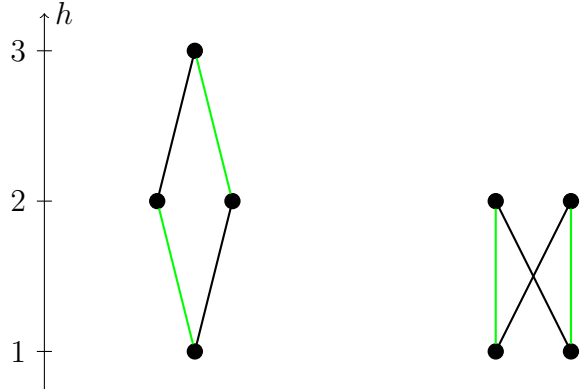
\begin{center}
\begin{tikzpicture}[scale=0.1]
\SetVertexSimple[MinSize=5pt]
\SetUpEdge[labelstyle={draw}]
\Vertex[x=0,y=0]{00}
\Vertex[x=0,y=40]{11}
\Vertex[x=-5,y=20]{10}
\Vertex[x=5,y=20]{01}
\Edge[color=green](00)(10)
\Edge[color=green](01)(11)
\Edge[color=black](00)(01)
\Edge[color=black](10)(11)

\Vertex[x=40,y=0]{a00}
\Vertex[x=50,y=0]{a11}
\Vertex[x=40,y=20]{a10}
\Vertex[x=50,y=20]{a01}
\Edge[color=green](a00)(a10)
\Edge[color=green](a01)(a11)
\Edge[color=black](a00)(a01)
\Edge[color=black](a10)(a11)

\draw [->] (-20,-5) -- (-20,45);
\draw (-19,0) -- (-21,0) node [align=right, left] {$1$};
\draw (-19,20) -- (-21,20) node [align=right, left] {$2$};
\draw (-19,40) -- (-21,40) node [align=right, left] {$3$};
\node [right] at (-20,45) {$h$};

\end{tikzpicture}
\end{center}

However, the height function must be invariant under $\pi_L(C)=\pi_R(C)=\langle 11\rangle$. In other words, moving once with an edge of both colors should not change the grading. Thus, both $A_1$ and $A_2$ can only be graded via the rank function depicted on the right. So they must look like this, after assigning colors:

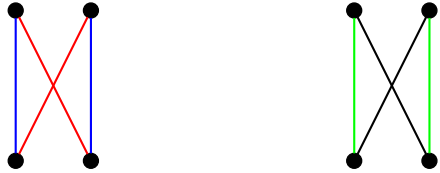
\begin{center}
\begin{tikzpicture}[scale=0.1]
\SetVertexSimple[MinSize=5pt]
\SetUpEdge[labelstyle={draw}]
\Vertex[x=-5,y=0]{00}
\Vertex[x=5,y=0]{11}
\Vertex[x=-5,y=20]{10}
\Vertex[x=5,y=20]{01}
\Edge[color=blue](00)(10)
\Edge[color=blue](01)(11)
\Edge[color=red](00)(01)
\Edge[color=red](10)(11)

\Vertex[x=40,y=0]{a00}
\Vertex[x=50,y=0]{a11}
\Vertex[x=40,y=20]{a10}
\Vertex[x=50,y=20]{a01}
\Edge[color=green](a00)(a10)
\Edge[color=green](a01)(a11)
\Edge[color=black](a00)(a01)
\Edge[color=black](a10)(a11)

\end{tikzpicture}
\end{center}

We now take the product $A_1\times A_2$, which has $4 \times 4 = 16$ vertices:
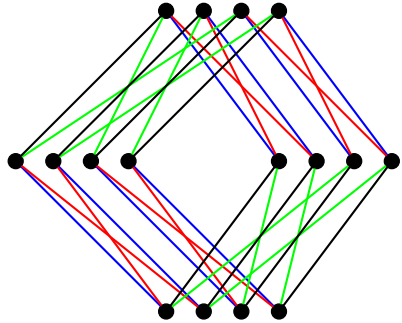
\begin{center}
\begin{tikzpicture}[scale=0.05]
\SetVertexSimple[MinSize=5pt]
\SetUpEdge[labelstyle={draw}]
\Vertex[x=0,y=0]{0000}
\Vertex[x=10,y=0]{0110}
\Vertex[x=20,y=0]{1001}
\Vertex[x=30,y=0]{1111}

\Vertex[x=-40,y=40]{0010}
\Vertex[x=-30,y=40]{0100}
\Vertex[x=-20,y=40]{1011}
\Vertex[x=-10,y=40]{1101}

\Vertex[x=30,y=40]{0001}
\Vertex[x=40,y=40]{0111}
\Vertex[x=50,y=40]{1000}
\Vertex[x=60,y=40]{1110}

\Vertex[x=0,y=80]{0011}
\Vertex[x=10,y=80]{0101}
\Vertex[x=20,y=80]{1010}
\Vertex[x=30,y=80]{1100}

\Edge[color=blue](0000)(0010)
\Edge[color=blue](0100)(0110)
\Edge[color=blue](1000)(1010)
\Edge[color=blue](1100)(1110)
\Edge[color=blue](0001)(0011)
\Edge[color=blue](0101)(0111)
\Edge[color=blue](1001)(1011)
\Edge[color=blue](1101)(1111)

\Edge[color=red](0000)(0100)
\Edge[color=red](0010)(0110)
\Edge[color=red](1000)(1100)
\Edge[color=red](1010)(1110)
\Edge[color=red](0001)(0101)
\Edge[color=red](0011)(0111)
\Edge[color=red](1001)(1101)
\Edge[color=red](1011)(1111)

\Edge[color=green](0000)(1000)
\Edge[color=green](0100)(1100)
\Edge[color=green](1110)(0110)
\Edge[color=green](1010)(0010)
\Edge[color=green](0001)(1001)
\Edge[color=green](0101)(1101)
\Edge[color=green](1111)(0111)
\Edge[color=green](1011)(0011)

\Edge[color=black](0000)(0001)
\Edge[color=black](0010)(0011)
\Edge[color=black](1000)(1001)
\Edge[color=black](1010)(1011)
\Edge[color=black](0111)(0110)
\Edge[color=black](0101)(0100)
\Edge[color=black](1111)(1110)
\Edge[color=black](1101)(1100)

\end{tikzpicture}
\end{center}

To construct $A_1 \times A_2/K$, recall that $K$ is generated by $\begin{bmatrix} 1 & 1 & 1 & 1\end{bmatrix}$, which means that each vertex is identified with the vertex that is obtained by following all four colors once. We should now get a graph with $16 / 2 = 8$ vertices. If we put an admissible dashing on it, we obtain a complete $2$-d Adinkra. One such choice recovers our Adinkra from Figure~\ref{fig:2d-example}:
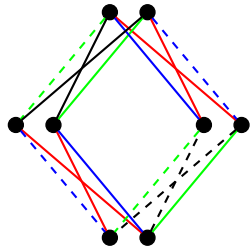
\begin{center}
\begin{tikzpicture}[scale=0.05]
\SetVertexSimple[MinSize=5pt]
\SetUpEdge[labelstyle={draw}]
\Vertex[x=0,y=0]{A}
\Vertex[x=-10,y=0]{H}
\Vertex[x=-35,y=30]{C}
\Vertex[x=-25,y=30]{B}
\Vertex[x=25,y=30]{D}
\Vertex[x=15,y=30]{E}
\Vertex[x=0,y=60]{G}
\Vertex[x=-10,y=60]{F}
\Edge[color=red](A)(C)
\Edge[color=red](B)(H)
\Edge[color=red](G)(E)
\Edge[color=red](F)(D)
\Edge[color=green](A)(D)
\Edge[color=green, style=dashed](E)(H)
\Edge[color=green](G)(B)
\Edge[color=green, style=dashed](F)(C)
\Edge[color=blue, style=dashed](C)(H)
\Edge[color=blue](B)(A)
\Edge[color=blue, style=dashed](G)(D)
\Edge[color=blue](F)(E)
\Edge[color=black, style=dashed](D)(H)
\Edge[color=black, style=dashed](A)(E)
\Edge[color=black](G)(C)
\Edge[color=black](B)(F)
\end{tikzpicture}
\end{center}
\end{ex}

\section{Conclusion and Future Work}
\label{sec:conclusion}
In this work, we have continued in the vein of \cite{hubsch:weaving} to study $2$-d Adinkras and provided stuctural results to study them combinatorially.  Describing already-known worldsheet supermultiplets in these terms could lead to new insights about these supermultiplets, and lead to the discovery of new worldsheet supermultiplets.

One of the motivations for \cite{d2l:first} was the idea of studying $4$-dimensional or higher-dimensional supermultiplets dimensionally reduced to $1$ dimension.  Likewise, the present work allows us to study the dimensional reduction to $2$ dimensions, which might carry important information about the original supermultiplet.

As in \cite{hubsch:continuum}, we could also consider non-adinkraic worldsheet supermultiplets.  It would be interesting to see if, as in the case of one dimension, there is a continuum of worldsheet supermultiplets.  Another direction to extend these results is local supersymmetry, as in the case of supergravity or superconformal theories, which could be of importance to superstrings.

Of course, one obvious step is to go to three or four dimensions.  It is expected that $SO(1,d-1)$ representations will play a role, and this work may begin to intersect \cite{faux:dimensional_enhancement,faux:spin_holography}.\footnote{Higher dimensions is also a natural place to consider gauge fields, which have not yet played a role in this discussion.  In one dimension gauge fields can be gauged away to zero.  In two dimensions this is not the case, but the corresponding field strengths automatically satisfy Bianchi identities, and so the Adinkra formalism works well in this case.  But in higher dimensions, gauge fields will be more difficult to avoid.} Beyond this, it is hoped that this work will generally help develop our knowledge of supersymmetry in two dimensions.

Many mathematicians may appreciate Adinkras simply as nice combinatorial objects with lots of structure, with surprising links to coding theory, switching graphs, and graph coloring. This view presents additional questions, less relevant to the physics but still mathematically interesting, in the spirit of Theorem~\ref{thm:esdeclassify}:
\begin{itemize} 
\item For example, it would be good to know how many ESDE's are ``compatible'' with a $1$-d Adinkra $A$ (i.e. there is some $2$-d Adinkra $A'$ with the same underlying graph and rankings as $A$, with the required ESDE as its code), or with the family of $1$-d Adinkras with the same underlying graph. 
\item  Enumerating all $2$-d Adinkras is also a natural goal, though fairly ambitious\footnote{We still do not know how to count $1$-d Adinkras with $I^n$ as underlying graph beyond small $n$. For $I^n/C$ with nontrivial $C$, we have almost no data! The work in \cite{zhang:adinkras} basically settles dashings completely and gives some structural results on rankings, but counting rankings completely remains a very difficult problem, related to the chromatic polynomial for some families of $C$.}, but counting all $2$-d Adinkras under some natural constraints may be fruitful. 
\item How often does the main theorem require no additional vertex switching (i.e. $F$ is the identity)? In general, when given a dashing, we can ask related questions about the minimum number of vertex switches needed to produce a dashing with a well-defined quotient.
\end{itemize}

\section*{Acknowledgments}
The authors wish to thank Charles Doran, Sylvester Gates, and Tristan H\"ubsch for helpful conversation.

\appendix

\section{Relation to H\"ubsch's Original Language}
\label{app:repn}
The statement of Theorem~\ref{thm:quotient} is a bit different from the statement of the conjecture in \cite{hubsch:weaving}.  There, the language was partly in terms of representations of the $2$-d SUSY algebra instead of graphs.
\begin{construction}[Construction 2.1 (off-shell)] Let $R_+$ and $R_-$ denote off-shell representations of two copies of the $1$-d SUSY algebra with $p$ and $q$ colors, respectively, and let $Z$ be a symmetry of $R_+ \otimes R_-$, as a representation. The $Z$-quotient of the tensor product $(R_+ \otimes R_-)/Z$ is then an off-shell representation of $2$-d SUSY algebra with $(p,q)$ colors.
\end{construction}
The conjecture in \cite{hubsch:weaving} then says that every Adinkraic representation of $2$-d SUSY with $(p,q)$ colors is obtained by this construction.

The relation between this conjecture and Theorem~\ref{thm:quotient} will be apparent once we establish the following relationships.  None of these facts are new to this paper, but this information is collected here for the convenience of the reader.
\begin{itemize}
\item The relationship between off-shell representations of SUSY and Adinkras: This is the central idea behind the original paper on Adinkras\cite{d2l:first}, and so we do not go into detail here.  The idea is that each vertex $v$ of the Adinkra corresponds to a field $f_v$ in the SUSY representation, and $Q_i$ acts on fields by the edge of color $i$, with possible derivatives depending on the grading or bigrading (as the case may be), and with an extra minus sign if the corresponding edge is dashed.
\item The relationship between $A_1\times A_2$ and $R_+\otimes R_-$: The definition of $R_+\otimes R_-$ as a representation of $2$-d SUSY that
\[Q_i(f_v\otimes f_w)=Q_i(f_v)\otimes f_w \]
if $i\le p$ and
\[Q_i(f_v\otimes f_w)=(-1)^{|h(v)|}f_v\otimes Q_i(f_w)\]
if $i >p$.  This is the standard way in which tensor products are defined in $\ZZ_2$-graded algebras.\cite{bott_tu,freed}  Then Construction~\ref{const:product} mimics this definition on the level of Adinkras.
\item A vertex switching at $v$ corresponds to replacing $f_v$ with $-f_v$.  Then all equations involving $Q_i(f_v)$ or $f_v$ will get an extra minus sign.  This reverses all dashings on edges connected to $v$.
\item A quotient defined in Theorem~\ref{thm:quotient} is a symmetry $Z$ of the representation.
\end{itemize}

\bibliographystyle{abbrv}
\bibliography{Adinkras}

\begin{thebibliography}{10}

\bibitem{bott_tu}
R.~Bott and L.~W. Tu.
\newblock {\em {Differential forms in algebraic topology}}, volume~82 of {\em
  Graduate Texts in Mathematics}.
\newblock Springer-Verlag, New York, 1982.

\bibitem{d2l:graph-theoretic}
C.~F. Doran, M.~G. Faux, S.~J. {Gates Jr.}, T.~H\"ubsch, K.~Iga, and G.~D.
  Landweber.
\newblock {On graph-theoretic identifications of Adinkras, supersymmetry
  representations and superfields}.
\newblock {\em International Journal of Modern Physics}, 22(5):869--930, 2007.

\bibitem{d2l:decodes}
C.~F. Doran, M.~G. Faux, S.~J. {Gates Jr.}, T.~H\"ubsch, K.~Iga, and G.~D.
  Landweber.
\newblock {Relating Doubly-Even Error-Correcting Codes, Graphs, and Irreducible
  Representations of N-Extended Supersymmetry}.
\newblock {\em Discrete and Computational Mathematics}, 2008.

\bibitem{d2l:omni}
C.~F. Doran, M.~G. Faux, S.~J. {Gates Jr.}, T.~H\"ubsch, K.~Iga, G.~D.
  Landweber, and R.~Miller.
\newblock {Codes and Supersymmetry in One Dimension}.
\newblock {\em Adv. Theor. Math. Phys.}, 15:1909--1970, 2011.

\bibitem{d2l:topology}
C.~F. Doran, M.~G. Faux, S.~J. {Gates Jr.}, T.~H\"ubsch, K.~Iga, G.~D.
  Landweber, and R.~L. Miller.
\newblock {Topology types of Adinkras and the corresponding representations of
  N-extended supersymmetry}.
\newblock {\em arXiv:0806.0050}, 2008.

\bibitem{dil:cohomology}
C.~F. Doran, K.~Iga, and G.~D. Landweber.
\newblock {An application of Cubical Cohomology to Adinkras and Supersymmetry
  Representations}.
\newblock \url{arXiv:1207.6806}, 2012.

\bibitem{douglas}
B.~L. Douglas, S.~J. {Gates Jr.}, B.~L. Segler, and J.~Wang.
\newblock {Automorphism Properties and Classification of Adinkras}.
\newblock {\em Advances in Mathematical Physics}, accepted.

\bibitem{d2l:first}
M.~G. Faux and S.~J. {Gates Jr}.
\newblock {Adinkras: A graphical technology for supersymmetric representation
  theory}.
\newblock {\em Physical Review D}, 71(6), 2005.

\bibitem{faux:dimensional_enhancement}
M.~G. Faux, K.~Iga, and G.~D. Landweber.
\newblock {Dimensional enhancement via supersymmetry}, July 2011.

\bibitem{faux:spin_holography}
M.~G. Faux and G.~D. Landweber.
\newblock {Spin holography via dimensional enhancement}.
\newblock {\em Physics Letters B}, 681:161--165, Oct. 2009.

\bibitem{freed}
D.~Freed.
\newblock {\em {Five lectures on supersymmetry}}.
\newblock American Mathematical Society, Providence, RI, 1999.

\bibitem{gates:dimensional_extension}
S.~J. {Gates Jr.} and T.~H\"ubsch.
\newblock {On Dimensional Extension of Supersymmetry: From Worldlines to
  Worldsheets}.
\newblock {\em Adv. in Th. Math. Phys.}, 16:1619--1667, 2012.

\bibitem{hubsch:weaving}
T.~H\"ubsch.
\newblock {Weaving Worldsheet Supermultiplets from the Worldlines Within}.
\newblock {\em Adv. in Th. Math. Phys.}, 17:1--72, 2013.

\bibitem{hubsch:continuum}
T.~H\"ubsch and G.~A. Katona.
\newblock A $q$-continuum of off-shell supermultiplets.
\newblock \url{arXiv:1310.3256}, 2013.

\bibitem{seidel:survey}
J.~J. Seidel.
\newblock {A survey of two-graphs}.
\newblock In {\em Colloquio Internazionale sulle Teorie Combinatorie (Rome,
  1973), Tomo I}, pages 481----511. Atti dei Convegni Lincei, No. 17. Accad.
  Naz. Lincei, Rome, 1976.

\bibitem{zhang:adinkras}
Y.~X. Zhang.
\newblock Adinkras for mathematicians.
\newblock {\em Transactions of the American Mathematical Society},
  366(6):3325--3355, 2014.

\end{thebibliography}

\end{document}